\documentclass[runningheads,envcountsame]{llncs}

\newcommand{\longversion}[1]{}
\newcommand{\shortversion}[1]{#1}
\usepackage{todonotes}
\usepackage{comment}

\usepackage[utf8]{inputenc}          %
\usepackage[T1]{fontenc}             %
\usepackage[USenglish]{babel}          %

\usepackage[intlimits]{amsmath}      %
\usepackage{amsfonts}                %
\usepackage{amssymb}                 %
\usepackage{amsmath}                 %

\shortversion{\usepackage{apxproof}
} 
\longversion{\usepackage[appendix=inline]{apxproof}}
\shortversion{\usepackage{apxproof}}%
\usepackage{mathtools}
\usepackage{setspace} 
\usepackage{graphicx}
\usepackage{hyperref}
\usepackage{cleveref}
\usepackage{xcolor}
\usepackage{xspace}
\usepackage{amsthm}

\crefname{corollary}{Corollary}{Corollary}
\crefname{lemma}{Lemma}{Lemma}
\crefname{claim}{Claim}{Claim}

\usepackage{caption}
\usepackage{subcaption}
\usepackage{tikz}
\usetikzlibrary{shapes.geometric}
\usepackage{orcidlink}

\newtheoremrep{thm}[theorem]{Theorem}
\newtheoremrep{lem}[theorem]{Lemma}
\newtheoremrep{obs}[theorem]{Observation}
\newtheoremrep{cor}[theorem]{Corollary}

\usepackage{algorithm}
\usepackage[noend]{algpseudocode}

\usepackage[most]{tcolorbox}
\newcommand{\problemdef}[4]{\setlength\tabcolsep{2pt}
	\begin{tcolorbox}[width = \columnwidth,%
    colback=gray!5!white,colframe=gray!75!black,
    arc=2pt,outer arc=3pt,boxrule=1.5pt,left =0.5em,right=0.5em,enhanced,attach boxed title to top center={yshift=-3.8mm,yshifttext=-.9mm},
    colbacktitle=gray!60,
  title=\textsc{#1} #2,
  boxed title style={size=small,%
  colframe=gray!50!black,left=0mm,right=0mm}]
  \mbox{}\\[-.5ex]
		\begin{tabular}{ @{\!\!\!}r p{0.89\columnwidth} c }
			\textsf{\bfseries Input:} & #3 \\[.5pt]
			\textsf{\bfseries Problem:} & #4
		\end{tabular}
	\vspace{-0.55em}
	\end{tcolorbox}
}

\newcommand{\NP}{\ensuremath{\mathsf{NP}}\xspace}
\newcommand{\FPT}{\ensuremath{\mathsf{FPT}}\xspace}
\newcommand{\XP}{\ensuremath{\mathsf{XP}}\xspace}

\newcommand{\OutputP}{\ensuremath{\mathsf{OutputP}}\xspace}

\newcommand{\DelayP}{\ensuremath{\mathsf{DelayP}}\xspace}
\newcommand{\DelayFPT}{\ensuremath{\mathsf{DelayFPT}}\xspace}

\newcommand{\OutputFPT}{\ensuremath{\mathsf{OutputFPT}}\xspace}

\newcommand{\W}[1]{\ensuremath{\mathsf{W}[#1]}\xspace}
\newcommand{\paraNP}{\ensuremath{\mathsf{para\text-NP}}\xspace}

\newcommand{\ptime}{\ensuremath{\mathsf{P}}\xspace}

\newcommand{\iffl}{if\longversion{ and only i}f }
\newcommand{\ExtGMDefAll}{\longversion{\textsc{Extension inclusion Minimal Defensive Alliance}}\shortversion{{\normalfont\textsc{ExtIMDefAll}}}\xspace}

\newcommand{\no}{\textsf{no}\xspace}
\newcommand{\yes}{\textsf{yes}\xspace}

\newcommand{\ilp}{\longversion{integer linear program}\shortversion{ILP}\xspace}
\newcommand{\ilps}{\longversion{integer linear program}\shortversion{ILPs}\xspace}

\newcommand{\nd}{\ensuremath{\mathbf{nd}}}

\newcommand{\tw}{\ensuremath{\mathbf{tw}}}
\newcommand{\pw}{\ensuremath{\mathbf{pw}}}
\newcommand{\pthree}{\ensuremath{\mathbf{p3}}}

\newcommand{\enum}[1]{{\normalfont\textsc{Enum}}\text-#1}
\newcommand{\AnIMDefAll}{{\normalfont\textsc{AnIMDefAll}}\xspace}
\newcommand{\MinMaxOut}{{\normalfont\textsc{MinMaxOut}}\xspace}
\newcommand{\twAnIMDefAll}{\ensuremath{\tw\text-\AnIMDefAll}\xspace}
\newcommand{\DefAll}{{\normalfont\textsc{DefAll}}}

\newcommand{\claimqed}{\hfill$\Diamond$\par}
\newenvironment{pfclaim}{{\noindent\it Proof:}}{\claimqed}

\makeatletter
\newcommand{\crossout}[1]{%
  \begingroup
  \sbox\z@{#1}%
  \dimen\z@=\wd\z@
  \dimen\tw@=\ht\z@
  \dimen\z@=.99626\dimen\z@   %
  \dimen\tw@=.99626\dimen\tw@ %
  \edef\co@wd{\strip@pt\dimen\z@}%
  \edef\co@ht{\strip@pt\dimen\tw@}%
  \leavevmode
  \rlap{\pdfliteral{q 1 J 0.4 w 0 0 m \co@wd\space \co@ht\space l S Q}}%
  \rlap{\pdfliteral{q 1 J 0.4 w 0 \co@ht\space m \co@wd\space 0 l S Q}}%
  #1%
  \endgroup
}
\makeatother

\newif\ifinappendix
\newcommand{\appref}[1]{\shortversion{{$(*)$}}}
\newcommand{\applabel}[1]{
\shortversion{{$(*)$}}
\label{#1}}

\title{Parameterized Enumeration of Minimal Defensive Alliances}
\titlerunning{Enumerating Defensive Alliances}
\author{Henning Fernau\inst1\orcidID{0000-0002-4444-3220} 
 \and Kevin Mann\inst1\orcidID{0000-0002-0880-2513}\and Arne Meier\inst2\orcidID{0000-0002-8061-5376}\and \\Heribert Vollmer \inst2\orcidID{0000-0002-9292-1960} }

\authorrunning{H. Fernau, K. Mann, A. Meier, H. Vollmer}
\institute{
 Universit\"at Trier, FB~4 -- Informatikwissenschaften, 54286 Trier, Germany.\\
\email{\{fernau,mann\}@uni-trier.de}\\
\and Institut für Theoretische Informatik, Leibniz Universität Hannover,\\
Hannover, Germany 
\email{\{meier,vollmer\}@thi.uni-hannover.de} 
}

\begin{document}
\maketitle

\begin{abstract}
    In this paper, we consider the complexity of enumerating inclusion minimal defensive alliances. We present a polynomial-delay algorithm on graphs with maximum degree 5. 
    We complement this result by proving that there is no output-polynomial algorithm on bipartite graphs with maximum degree 6 and degeneracy 2, unless $\ptime = \NP$. 
    Furthermore, there is an \FPT-delay algorithm when parameterized by neighborhood diversity. This result is the first exploit of a recently published enumeration algorithm for ILPs.
    By way in contrast, we prove that, for the parameter pathwidth, there is no \FPT-delay algorithm for enumerating all inclusion minimal defensive alliances (unless $\FPT = \W{1}$). 
    T\longversion{o the best of our knowledge, t}his is the first non-enumerability result using parameterized complexity for variations of \textsc{Another/Next} problems. 
\end{abstract}

\section{Introduction}

\paragraph{Alliances in Graphs.}
The study of ``alliances'' in graph theory is well established and dates back to the paper by Kristiansen, Hedetniemi, and Hedetniemi~\cite{KriHedHed2004} in 2004. A \emph{defensive alliance} in an undirected graph is a non-empty subset of vertices $D$, where every vertex in $D$ has at least as many allies (neighbors in $D$ counting itself) as it has enemies (neighbors outside $D$).
Formally, let $G=(V,E)$ be an undirected graph, then 
$\emptyset\subsetneq D\subseteq V$ is a {defensive alliance} if, for each $v\in D$, $\deg_D(v)+1\geq \deg_{\overline{D}}(v)$, where $\deg_X(v)$ is the number of neighbors of $v$ in~$X\subseteq V$. Clearly, $\deg_{\overline{D}}(v)=\deg(v)-\deg_D(v)$, where $\deg(v)$ denotes the \emph{degree} of~$v$, i.e., the size of the open neighborhood $N(v)=\{u\in V\mid \{v,u\}\in E\}$. The open neighborhood of a vertex set $A\subseteq V$ is defined as $N(A)=\bigcup_{a\in A} N(a)$.
A graph~$G$ has \emph{degeneracy} at most~$d$ if all the vertices of $G$ can be ordered as $v_1<v_2<\cdots<v_n$ such that $\deg_{\{v_j\mid i< j\leq n\}}(v_i)\leq d$ for all $i\in [n]$.
Notice that a singleton set forms a defensive alliance \iffl it comprises of a \emph{pendant vertex}, i.e., a vertex of degree~1, or of an \emph{isolated vertex}, i.e., a vertex of degree~0. 
Also, inclusion minimal defensive alliances ($D$ is a defensive alliance and each $U\subsetneq D$ is not defensive alliance) must be connected\longversion{ as each connected component of a defensive alliance is a defensive alliance}. 
Moreover, there is a polynomial-time algorithm to determine whether a defensive alliance~$D$ is inclusion minimal or not~\cite[Prop.~5]{BazFerTuz2019}. 
This is not completely obvious, as the property of being a defensive alliance is not monotone, and this is why one has to differentiate between inclusion minimality (also called global minimality in the literature) and local minimality, where it is only required for a defensive alliance~$D$ that, for each $x\in D$, $D\setminus\{x\}$ is no defensive alliance.

\begin{example}\label{exa-K4+pendant}
    Consider a graph formed by a $K_4$ with vertices $v_1,v_2,v_3,v_4$ plus a pendant vertex $v_0$ attached to~$v_1$. Then, $\{v_0\}$ is an inclusion minimal defensive alliance, while $\{v_0,v_1,v_2\}$ is a locally minimal defensive alliance that is not inclusion minimal.
\end{example}

\paragraph{Previous Results.}
To the best of our knowledge, the work of Feng~\emph{et~al.}~\cite{FenFerMan2026} includes the only published results for enumerating inclusion minimal defensive alliances; some of the results are on input-sensitive enumeration (running time is bounded by the input size) and others on output-sensitive results (running time is bounded by both input and   output size); for example, there is an optimal algorithm that enumerates all inclusion minimal defensive alliances in time $\mathcal{O}^*\!\left( 2^n\right)$. 
We restate two central results\longversion{ for the ease of reference}.
\begin{theorem}[Theorem 4 of \cite{FenFerMan2026}] \label{thm:enum-forests-defall}
   All inclusion minimal defensive alliances of trees of order $n$ can be enumerated in time $\mathcal{O}^*\!\left( \sqrt{2}^n\right)$ with polynomial delay.
\end{theorem}

\begin{theorem}[Theorems 5 and 6 of \cite{FenFerMan2026}]\label{thm:outpolyhard_bipartite}
If there was an output-polynomial time algorithm for enumerating inclusion minimal
defensive alliances of bipartite or split graphs, then there would be an output-polynomial time algorithm for enumerating minimal dominating sets in general connected graphs. 
\end{theorem}

\paragraph{Contributions.} In this paper, we will strengthen these two results. \longversion{Clearly, trees have a degeneracy of $1$ and forests are exactly the graphs with degeneracy~$1$. }By \autoref{thm:enum-forests-defall}, we have a polynomial-delay enumeration algorithm for inclusion minimal defensive alliances on graphs of degeneracy~$1$. 
We will prove that there is  no output-polynomial enumeration algorithm on bipartite graphs of degeneracy~$2$ and degree \longversion{at most~$6$}\shortversion{$\le 6$} unless $\ptime=\NP$ (see \autoref{thm:no_polydelay_max_degree_6}). 
This result is complemented by a polynomial-delay enumeration algorithm on graphs of degree at most~5 (see \autoref{thm:poly-delay-for-maxdegree-5}). We also commence a parameterized analysis of this enumeration problem, providing an \FPT-delay enumeration algorithm (with the parameter neighborhood diversity; see \autoref{thm:FPT-delay_nd}) and ruling out even an output-\FPT algorithm with the parameter pathwidth in \autoref{thm:no_pw_FPT_enum}. 

\paragraph{Organization.} After presenting preliminaries in \Cref{sec:prelims}, we consider polynomial-delay enumeration for inclusion minimal defensive alliances or graphs with maximum degree at most 5 in \Cref{sec:enum-max-degree5}. 
Then, in \Cref{sec:lowerbound}, we provide the lower bounds. 
Eventually, in \Cref{sec:paraenum}, we consider parameterized enumeration algorithms for the parameters neighborhood diversity, pathwidth and treewidth. 
\shortversion{Due to space constraints, omitted proofs of results are marked with an asterisk $(*)$. We refer to the appendix for further information.}

\section{Preliminaries}\label{sec:prelims}
\paragraph{Notation.} Throughout the paper, we use mathematical  standard notation. In particular, $\mathbb{N}$ ($\mathbb{N}^+$, resp.) is the set of non-negative (or positive) integers and for $k\in\mathbb{N}^+$, $[k]=\{i\in\mathbb{N}^+\mid i\leq k\}$. We also assume basic (parameterized) complexity knowledge on the side of the reader~\cite{CygFKLMPPS2015,DowFel99,DowFel2013,Pap94}.

\paragraph{Enumeration.}
Let $\Sigma$ be an alphabet and $A \subseteq \Sigma^* \times \Sigma^*$. For $x\in \Sigma^*$, $A(x)=\{\,y\in \Sigma^* \mid (x,y) \in A\,\}$. In the following, we assume that $A$ is \emph{polynomially balanced}, i.e., there is a polynomial $q$ such that, for all $(x,y)\in A$, $|y|\leq q(|x|)$. 
For the enumeration problem $\enum A$, an $x\in \Sigma^*$ is given and the task is to output all elements in $A(x)$ without duplicates.

In our paper, $A$ will be $\DefAll$, consisting of (codings of) graphs and inclusion minimal defensive alliances. Hence, given a graph~$G$, the problem $\enum \DefAll$ asks to list all inclusion minimal defensive alliances of~$G$. For instance, reconsider \autoref{exa-K4+pendant}. Apart from $\{v_0\}$, $\{v_2,v_3\},\{v_2,v_4\},\{v_3,v_4\}$ should be listed.

An enumeration algorithm $\mathcal{A}$ runs in \emph{output-polynomial time} if there is a polynomial $p\colon \mathbb{N} \times \mathbb{N} \to \mathbb{N}$  such that, for all $x\in \Sigma^*$, the algorithm $\mathcal{A}$ outputs all elements from $A(x)$ in at most $p(\vert x\vert, \vert A(x) \vert )$ steps. 
The class $\OutputP$ is the set of all enumeration problems for which there is an output-polynomial algorithm. 

An enumeration algorithm $\mathcal{A}$ runs with \emph{polynomial delay} if there is a polynomial $p\colon \mathbb{N} \to \mathbb{N}$  such that, for all $x\in \Sigma^*$, the algorithm $\mathcal{A}$ outputs all elements from $A(x)$ and the time between any two consecutive outputs---from the start until the first output as well as from the last output until the algorithm stops---is bounded by $p(\vert x\vert)$. 
The class $\DelayP$ consists of all enumeration problems for which there is a polynomial-delay algorithm. 
By definition, we have\longversion{ the following inclusion landscape}: $\DelayP%
\subseteq \OutputP$.

For some enumeration problems, the polynomial features just defined seem to be out of reach. In analogy to the case of decision problems, the idea of parameterized algorithms was developed~\cite{CreMMSV2017,Mei2020}.
To this end, let $\kappa\colon\Sigma^*\to\mathbb{N}$ be a \emph{parameterization function}.\longversion{\footnote{In~\cite{CreMMSV2017}, it was required that $\kappa$ is polynomial-time computable. As this requirement is not existent with decision problems and will not be satisfied for some parameterizations that we consider in this paper, we do not put this condition here.}} 
Then, we say that an enumeration algorithm $\mathcal{A}$ runs in \emph{output-\FPT time} (with respect to the parameterization~$\kappa$) if there is a polynomial $p\colon \mathbb{N} \times \mathbb{N} \to \mathbb{N}$ and some computable function $f\colon\mathbb{N}\to \mathbb{N}$ such that, for all $x\in \Sigma^*$, the algorithm $\mathcal{A}$ outputs all elements from $A(x)$ in at most $f(\kappa(x))\cdot p(\vert x\vert, \vert A(x) \vert )$ steps. $\OutputFPT$ is the corresponding class of parameterized enumeration problems. Analogously, one can define the class %
 $\DelayFPT$.
For clarity, we prefix the parameterization function to the enumeration problem, leading to $\kappa$-$\enum A$.

\section{Enumeration with Maximum Degree at most 5}\label{sec:enum-max-degree5}

The main goal of this section is to prove that inclusion minimal defensive alliances can be enumerated with polynomial delay on graphs with maximum degree~5. To achieve this, we will first prove a connection to path and cycle enumeration. In the algorithm we will then use results of Uno and Satoh~\cite{UnoSat2014}.  

\begin{lemma}\label{lem:max_deg_5_path__min_def_all}
Let $G=(V,E)$ be a graph with $\Delta(G)\leq 5$ and $p_1, \ldots, p_{\ell}$ be an induced path with $\ell\geq 2$.  $P=\{p_1, \ldots, p_{\ell}\}\subseteq V$ is an inclusion minimal defensive alliance \iffl $\deg(p_1),\deg(p_{\ell})\in \{2,3\}$ and, for\longversion{ all} $i\in \{2,\ldots,  \ell-1\}$, $\deg(p_i)\in \{4,5\}$.  
\end{lemma}

\begin{proof}
    First we consider $\deg(p_1),\deg(p_{\ell})\in \{2,3\}$ and, for all $i\in \{2,\ldots,  \ell-1\}$, $\deg(p_i)\in \{4,5\}$. 
    For $i\in \{1,\ell\}$, $\deg_{P}(p_i)+1 \geq 2\geq \deg_{\overline{P}}(p_i)$. Furthermore, $\deg_{P}(p_i)+1 \geq 3\geq \deg_{\overline{P}}(p_i)$ for $i\in \{2\ldots, \ell-1\}$. Thus $P$ is a defensive alliance. \longversion{ 
}       
    Let $\emptyset \subsetneq Q \subsetneq P$. It is easy to see that $Q$ is no defensive alliance, as there is an $i\in \{1,\ell\}$ with $\deg_Q(p_i) + 1=1 <2 \leq \deg_{\overline{Q}}(p_i)$ or, there is an $i\in \{2,\ldots,\ell-1\}$ with $\deg_Q(p_i) + 1\leq 2 <3 \leq \deg_{\overline{Q}}(p_i)$.  Hence, $Q$ is no defensive alliance and $P$ is inclusion minimal.\longversion{ 
}    Conversely, assume $P$ is an inclusion minimal defensive alliance, but $\deg(p_1),\deg(p_{\ell})\in \{2,3\}$ and for all $i\in \{2,\ldots,  \ell-1\}$, $\deg(p_i)\in \{4,5\}$ does not hold. If there is a vertex in $P$ of degree one, $P$ is not minimal as $\ell\geq 2$. Therefore, we can assume $\deg(v)>1$ for all $v\in P$. If there is an $i\in \{1,\ell\}$ with $\deg(v)>3$, $\deg_P(v) + 1 = 2 < 3 \leq \deg_{\overline{P}}(v)$. Hence, there is an $i\in \{2,\ldots, \ell - 1 \}$ with $\deg(p_i)\in \{2,3\}$. In this case, $\{p_1,\ldots, p_i\}$ is a defensive alliance. \qed
\end{proof}

\begin{lemrep}\applabel{lem:max_deg_5_cycle__min_def_all}
Let $G=(V,E)$ be a graph with $\Delta(G)\leq 5$ and $p_1, \ldots, p_{\ell}$ be an induced cycle.  $P=\{p_1, \ldots, p_{\ell}\}\subseteq V$ is an inclusion minimal defensive alliance \iffl $\vert \{i\in [\ell]\mid \deg(p_i)<4\} \vert\leq 1$.
\end{lemrep}

\begin{proof}
    First we assume that the degree condition $\vert \{i\in [\ell]\mid \deg(p_i)<4\} \vert\leq 1$ holds. 
    For $i\in [\ell]$, $\deg_{P}(p_i)+1 \geq 3\geq \deg_{\overline{P}}(p_i)$. Thus, $P$ is a defensive alliance. 
    Let $\emptyset \subsetneq Q \subsetneq P$. By the degree condition, no singleton set~$Q$ can be a defensive alliance. Otherwise, if $Q$ were an inclusion minimal defensive alliance, $Q$ would form an induced path, so that the degree condition contradicts \autoref{lem:max_deg_5_path__min_def_all}.

    Conversely, let $P$ be an inclusion minimal defensive alliance. Since $P$ is a cycle, it contains no pendant vertex. If there are $i,j\in [\ell]$ with $i<j$ and $\deg(p_i),\deg(p_j)\in \{2,3\}$, but $\deg(p_k)\in\{4,5\}$ for $i<k<j$, then by \autoref{lem:max_deg_5_path__min_def_all}, $\{p_i,\ldots,p_j\}$ would be a defensive alliance, contradicting minimality.\qed
\end{proof}

\begin{lemma}\label{lem:max_deg_5_min_def_all}
Let $G=(V,E)$ be a graph with $\Delta(G)\leq 5$ and $A\subseteq V$. $A$ is an inclusion minimal defensive alliance \iffl one of the following constraints holds:
\begin{enumerate}
    \item $A=\{v\}$ with $\deg(v)\geq 1$.
    \item $\vert A\vert>1$ and $A$ forms an induced path $p_1,\ldots, p_{\ell}$ with $\deg(p_1),\deg(p_{\ell})\in \{2,3\}$ and for all $i\in \{2,\ldots,  \ell-1\}$, $\deg(p_i)\in \{4,5\}$.    
    \item $A$ forms an induced cycle $p_1,\ldots, p_{\ell}$ with $\vert \{i\in [\ell]\mid \deg(p_i)<4\} \vert\leq 1$.
\end{enumerate}
\end{lemma}

\begin{proof}
    The if-part follows by \autoref{lem:max_deg_5_path__min_def_all} and \autoref{lem:max_deg_5_cycle__min_def_all}. 

    Conversely, let $A \subseteq V$ be an inclusion minimal defensive alliance. %
    Assume $A$ contains an induced cycle $C \subsetneq A$. Since for all $v\in C$, $\deg_C(v)+1\geq 3\geq\deg(v)-\deg_C(v)= \deg_{\overline {C}}(v)$, $C$ is a defensive alliance. Thus by minimality, $G[A]$ is a tree. If $\vert A\vert = 1$, the element is pendant or isolated in~$G$. Now assume $\vert A\vert>1$. Consider some internal vertex~$x$ of the tree~$A$. Assume $\deg(x)\in \{2,3\}$. Let $T$ be a connected component of $G[A\setminus \{x\}]$. Since $\deg_A(v)= \deg_T(v)$ for any $v\in T$ and $\deg_T(x)+1=2\geq \deg_{\overline{T}}(x)$, $T\subseteq A$ is a defensive alliance, contradicting minimality. Hence, all internal vertices must have degrees~4 or~5.  Now, let $u,v\in A$ be pendant vertices in $A$. Observe that $\deg(u),\deg(v)\in \{2,3\}$, as otherwise $\deg_A(u)+1<\deg_{\overline{A}}(u)$ or $\deg_A(v)+1<\deg_{\overline{A}}(v)$. By \autoref{lem:max_deg_5_path__min_def_all}, the unique $u$-$v$-path $P$ in $G[A]$ is a defensive alliance of~$G$. This proves the only-if-part.\qed 
\end{proof}

\noindent
\autoref{lem:max_deg_5_min_def_all} is the core idea for the following result (using the algorithms of~\cite{UnoSat2014}). 
\begin{theorem}\label{thm:poly-delay-for-maxdegree-5}
    There is a polynomial-delay algorithm for enumerating all inclusion minimal defensive alliances for graphs $G$ with $\Delta(G)\leq 5$. 
\end{theorem}

\begin{proof} 
    \longversion{Uno and Satoh proved~\cite{UnoSat2014} that}\shortversion{According to~\cite{UnoSat2014}}, for a given graph $G=(V,E)$ and two vertices $s,t$, all $s$-$t$-paths as well as all cycles can be enumerated with polynomial delay. 
    
    \shortversion{With $V_{i}\coloneqq\{v\in V \mid \deg(v)\in \{i,i+1\}\}$, we denote the vertices with degree $i$ or $i+1$ for $i\in\{2,4\}$. The algorithm will enumerate the inclusion minimal defensive alliances depending on the three cases mentioned in \autoref{lem:max_deg_5_min_def_all}. Observe that the three cases are distinct. Therefore, we only need to find algorithms enumerating each of the three cases with polynomial delay. For the first case, the algorithm only needs to enumerate all vertices of degree~1. 
    In the second case, for each distinct pair $s,t\in V_2$, we enumerate all induced $s$-$t$-paths on $G[\{s,t\} \cup V_4]$. In this way we enumerate all paths for which the start and end points have degree~$2$ or~$3$ and the internal nodes have degree~$4$ or~$5$.
    This leaves us to consider the third case.
    Here we have to enumerate cycles $c_1,\ldots,c_{\ell},c_{\ell+1}=c_1$, divided into to the two cases ($V_2\cap \{c_1,\ldots,c_{\ell}\}\neq \emptyset$ and $V_2\cap \{c_1,\ldots,c_{\ell}\}= \emptyset$).
    We start with {$V_2\cap \{c_1,\ldots,c_{\ell}\}\neq \emptyset$}: The way we do this is inspired by \cite{ReaTar75,UnoSat2014}. By \autoref{lem:max_deg_5_min_def_all}, we know that there is at most one vertex $v\in V_2\cap \{c_1,\ldots,c_{\ell}\}$. The algorithm goes through all pairs of distinct vertices $s,t\in N(v)\cap V_4$ and enumerates all induced $s$-$t$-paths in $G[(V_4\setminus N(v))\cup \{s,t\}]$. Such an induced path $p_1,\ldots, p_\ell$ corresponds to a unique induced cycle including $v$, as $s=p_1,\ldots, p_\ell=t$ is an induced $s$-$t$-path (with $s,t\in N(v)$) iff  $v,s=p_1,\ldots, p_\ell=t,v$ is an induced cycle, as by construction, the only neighbors of $v$ on this cycle are $p_1$ and $p_\ell$. Observe that this procedure enumerates all cycles including exactly one vertex from $V_2$ and none twice.
    For the remaining defensive alliances in $G$, we only need to enumerate all induced cycles in $G[V_4]$ by using the algorithm in~\cite{UnoSat2014}.}
\qed 
\end{proof}

\section{Another Inclusion Minimal Defensive Alliance}\label{sec:lowerbound}

This section is dedicated to proving the following hardness result.

\begin{theorem}\label{thm:no_polydelay_max_degree_6}
   There is no output-polynomial time algorithm enumerating all inclusion minimal defensive alliances, even on bipartite graphs of degeneracy 2 and maximum degree 6, unless $\ptime=\NP$.
\end{theorem}

In view of the results from the previous section, one observes a sharp boundary between graphs of maximum degree 5 and~6, and likewise, between graphs of degeneracy~1 (i.e., forests), also see \autoref{thm:enum-forests-defall} for the enumeration case, and those with degeneracy 2 or larger.
There is an interesting parallel to the classical decision problem, i.e., the question if, given a graph~$G$ and an integer~$k$, there exists a defensive alliance of size at most~$k$ in~$G$: also here, we have an easy case, i.e., a  polynomial-time algorithm, for solving this problem on graphs of maximum degree~5, while the case of graphs of maximum degree~6 poses an \textsf{NP}-hard situation, see \cite{RedKar2026a}. 
To prove this theorem, we employ the following problem.
 
\problemdef{Another Inclusion Minimal Defensive Alliance} {(\AnIMDefAll)}{An undirected graph $G=(V,E)$, and a set $L \subseteq 2^V$ of inclusion minimal defensive alliances.}{Is there an inclusion minimal defensive alliance $D\subseteq V$ with $D \notin L$?}

This definition is based on \cite{KurMan2026}. A similar technique was used in \cite{BorMak2024,BroDKLUW2024,JohPapYan88a}.
Observe that there is a difference between this definition of \textsc{Another} problems and the one in \cite{CreKPSV2019,Str2010,Str2019} (in \cite{CreKPSV2019}, \textsc{Exists-Another} corresponds to the \textsc{Another} as defined in this paper and \textsc{Another} problems defined in \cite{CreKPSV2019} are functional problems). In our definition, we have a decision problem. So it would be enough if the output is a \yes or a \no. In the definition  \cite{CreKPSV2019,Str2010,Str2019}, for the case of a \yes, we need to produce a defensive alliance (i.e., this is a functional problem). This makes a difference, as  \cite{Str2010} includes results for \textsf{IncP}, while our definition gives results for \textsf{OutputP}. 

\begin{proposition}[{\cite[Lemma~4.3]{BokEMM2017}}]\label{prop:from-enumeration-to-decision}
If there is an output-polynomial time algorithm~$\mathcal{A}$ for enumerating inclusion minimal defensive alliances, then we can conclude that  $\AnIMDefAll\in\ptime$.
\end{proposition}
The idea is as follows:
Assume that there is an output-polynomial time algorithm~$\mathcal{A}$ for enumerating inclusion minimal defensive alliances, with the polynomial~$p$ bounding the running time. We can then construct a polynomial-time algorithm for solving \AnIMDefAll as follows.
Let $x$ be an instance and $L$ be the list of given solutions. Then we run $\mathcal{A}$ on~$x$. Every time $\mathcal{A}$ produces an output $D$, we check whether $D\in L$. If $D\notin L$ we can output \yes and terminate, as we have enumerated a solution that is not on our list. Otherwise, $\mathcal{A}$ continues. Assume that $\mathcal{A}$ needs would run more than $p(\vert x\vert, \vert L\vert +1)$ steps. Then, $\mathcal{A}$ would output at least $\vert L \vert + 1$ solutions. By pigeon hole, at least one of these solutions, say, $D$, would not belong to~$L$. Hence, upon seeing~$D$, our algorithm would have already output \yes and would have terminated already. Hence, we can safely terminate our algorithm after $p(\vert x\vert, \vert L\vert+1)$ steps. Namely, if we did not see a solution that is not on our list in this time, we would have enumerated the whole list~$L$ and nothing more and therefore, we can output \no.

We will now prove that \AnIMDefAll is \NP-complete. 
Together with \autoref{prop:from-enumeration-to-decision}, this shows \autoref{thm:no_polydelay_max_degree_6}.
For the $\NP$-hardness proof, we use \textsc{Monotone 3-Sat-(2,2)}. This is a variation of \textsc{3-Sat}, where each variable appears twice positively and twice negatively, and each clause either contains three non-negated variables or three negated variables. Darmann and Döcker~\cite{DarDoc2021} provide an \NP-completeness proof for this\longversion{ problem}.

    \begin{figure}[bt]
    	\centering
	\scalebox{.94}{\begin{tikzpicture}[transform shape]
		      \tikzset{every node/.style={circle,minimum size=0.1cm}}
			\node  (labelwi) at (0.4,0.4) {$z_{i}$};
            \draw[rounded corners] (6.5, 1.5) rectangle (-6.5, 4.5); 
			\node[draw,label={below:$c_{j_1}$},fill=black] (c1) at (4,1) {};
			\node[draw,label={above:$c_{j_2}$},fill=black] (c2) at (4,-1) {};
			\node[draw,label={below:$c_{j_3}$},fill=black] (c3) at (-4,1) {};
			\node[draw,label={above:$c_{j_4}$},fill=black] (c4) at (-4,-1) {};
			\node[draw,rectangle] (c'1) at (4.75,1) {};
			\node[draw,rectangle] (c'2) at (4.75,-1) {};
			\node[draw,rectangle] (c'3) at (-4.75,1) {};
			\node[draw,rectangle] (c'4) at (-4.75,-1) {};
			\node[draw,label={below:$y_i$},fill=black] (yi) at (1.3,0) {};
			\node[draw,label={below:$\overline{y}_i$},fill=black] (y'i) at (-1.3,0) {};
			\node[draw,label={below:$x_i$}] (xi) at (2.6,0) {};
			\node[draw,label={below:$\overline{x}_i$}] (x'i) at (-2.6,0) {};
			\node[draw,rectangle,fill=black] (xi1) at (2.3,0.5) {};
			\node[draw,rectangle,fill=black] (xi2) at (2.6,0.5) {};
			\node[draw,rectangle,fill=black] (xi3) at (2.9,0.5) {};
			\node[draw,rectangle,fill=black] (x'i1) at (-2.3,0.5) {};
			\node[draw,rectangle,fill=black] (x'i2) at (-2.6,0.5) {};
			\node[draw,rectangle,fill=black] (x'i3) at (-2.9,0.5) {};
			\node[draw,rectangle,fill=black,label={right:$w_{i}$}] (wi) at (0,-1) {};
			\node[draw] (zi) at (0,0) {};
			\node[draw] (u1) at (-4.9,2) {};
			\node[draw,fill=black] (u2) at (-4,2) {};
			\node[draw] (u3) at (-3.1,2) {};
			\node[draw] (u4) at (-0.9,2) {};
			\node[draw,fill=black] (u5) at (0,2) {};
			\node[draw] (u6) at (0.9,2) {};
			\node[draw] (u7) at (3.1,2) {};
			\node[draw,fill=black] (u8) at (4,2) {};
			\node[draw] (u9) at (4.9,2) {};
			\node[draw,rectangle,fill=black] (u11) at (-5.2,2.75) {};
			\node[draw,rectangle,fill=black] (u12) at (-4.9,2.75) {};
			\node[draw,rectangle,fill=black] (u13) at (-4.6,2.75) {};
			\node[draw,rectangle] (u21) at (-4.3,2.75) {};
			\node[draw,rectangle] (u22) at (-4,2.75) {};
			\node[draw,rectangle] (u23) at (-3.7,2.75) {};
			\node[draw,rectangle,fill=black] (u31) at (-3.4,2.75) {};
			\node[draw,rectangle,fill=black] (u32) at (-3.1,2.75) {};
			\node[draw,rectangle,fill=black] (u33) at (-2.8,2.75) {};
			\node[draw,rectangle,fill=black] (u41) at (-1.2,2.75) {};
			\node[draw,rectangle,fill=black] (u42) at (-0.9,2.75) {};
			\node[draw,rectangle,fill=black] (u43) at (-0.6,2.75) {};
			\node[draw,rectangle] (u51) at (-0.3,2.75) {};
			\node[draw,rectangle] (u52) at (0,2.75) {};
			\node[draw,rectangle] (u53) at (0.3,2.75) {};
			\node[draw,rectangle,fill=black] (u61) at (0.6,2.75) {};
			\node[draw,rectangle,fill=black] (u62) at (0.9,2.75) {};
			\node[draw,rectangle,fill=black] (u63) at (1.2,2.75) {};
			\node[draw,rectangle,fill=black] (u71) at (2.8,2.75) {};
			\node[draw,rectangle,fill=black] (u72) at (3.1,2.75) {};
			\node[draw,rectangle,fill=black] (u73) at (3.4,2.75) {};
			\node[draw,rectangle] (u81) at (3.7,2.75) {};
			\node[draw,rectangle] (u82) at (4,2.75) {};
			\node[draw,rectangle] (u83) at (4.3,2.75) {};
			\node[draw,rectangle,fill=black] (u91) at (4.6,2.75) {};
			\node[draw,rectangle,fill=black] (u92) at (4.9,2.75) {};
			\node[draw,rectangle,fill=black] (u93) at (5.2,2.75) {};
			\node[] (dots1) at (0,4) {$\cdots$};
			\node[] (dots2) at (-2,2) {$\cdots$};
			\node[] (dots3) at (2,2) {$\cdots$};
			\node[] (U) at (0,3.5) {$U$};
            
    		\path (zi) edge[-] (wi);
    		\path (zi) edge[-] (yi);
    		\path (zi) edge[-] (y'i);
    		\path (xi) edge[-] (yi);
    		\path (x'i) edge[-] (y'i);
    		\path (c1) edge[-] (xi);
    		\path (c2) edge[-] (xi);
    		\path (c3) edge[-] (x'i);
    		\path (c4) edge[-] (x'i);
    		\path (c1) edge[-] (c'1);
    		\path (c2) edge[-] (c'2);
    		\path (c3) edge[-] (c'3);
    		\path (c4) edge[-] (c'4);
    		\path (u5) edge[-] (zi);
    		\path (u1) edge[-] (c3);
    		\path (u3) edge[-] (c3);
    		\path (u7) edge[-] (c1);
    		\path (u9) edge[-] (c1);
    		\path (u1) edge[-] (u2);
    		\path (u2) edge[-] (u3);
    		\path (u4) edge[-] (u5);
    		\path (u5) edge[-] (u6);
    		\path (u7) edge[-] (u8);
    		\path (u8) edge[-] (u9);
    		\path (u1) edge[-] (u11);
    		\path (u1) edge[-] (u12);
    		\path (u1) edge[-] (u13);
    		\path (u2) edge[-] (u21);
    		\path (u2) edge[-] (u22);
    		\path (u2) edge[-] (u23);
    		\path (u3) edge[-] (u31);
    		\path (u3) edge[-] (u32);
    		\path (u3) edge[-] (u33);
    		\path (u4) edge[-] (u41);
    		\path (u4) edge[-] (u42);
    		\path (u4) edge[-] (u43);
    		\path (u5) edge[-] (u51);
    		\path (u5) edge[-] (u52);
    		\path (u5) edge[-] (u53);
    		\path (u6) edge[-] (u61);
    		\path (u6) edge[-] (u62);
    		\path (u6) edge[-] (u63);
    		\path (u6) edge[-] (u61);
    		\path (u6) edge[-] (u62);
    		\path (u6) edge[-] (u63);
    		\path (u7) edge[-] (u71);
    		\path (u7) edge[-] (u72);
    		\path (u7) edge[-] (u73);
    		\path (u8) edge[-] (u81);
    		\path (u8) edge[-] (u82);
    		\path (u8) edge[-] (u83);
    		\path (u9) edge[-] (u91);
    		\path (u9) edge[-] (u92);
    		\path (u9) edge[-] (u93);
    		\path (u3) edge[-] (dots2);
    		\path (u4) edge[-] (dots2);
    		\path (u6) edge[-] (dots3);
    		\path (u7) edge[-] (dots3);
    		\path (u9) edge[-] (6,2);
    		\path (u1) edge[-] (-6,2);
            \path (u9) edge[-] (6,2);
    		\path (-6,4) edge[-] (-6,2);
    		\path (6,4) edge[-] (6,2);
    		\path (6,4) edge[-] (dots1);
    		\path (-6,4) edge[-] (dots1);
            \path (xi) edge[-] (xi1);
            \path (xi) edge[-] (xi2);
            \path (xi) edge[-] (xi3);
            \path (x'i) edge[-] (x'i1);
            \path (x'i) edge[-] (x'i2);
            \path (x'i) edge[-] (x'i3);

    		\end{tikzpicture}}
        \caption{ \label{fig:trdf_paranp_maxdeg3}Sketch of the reduction using \textsc{Monotone 3-Sat-(2,2)}}

\end{figure} 

\begin{theorem}\label{thm:another_def_all_max_deg_6}
   \AnIMDefAll is $\NP$-complete, even on bipartite graphs of degeneracy~2 and maximum degree~6.
\end{theorem}

\begin{proof}
    The membership in $\NP$ follows by guess-and-check.
    
    For the $\NP$-hardness, we use \textsc{Monotone 3-Sat-(2,2)}. Let $X=\{x_1,\ldots,x_n\}$ be the set of variables and $C=\{c_1,\ldots,c_m\}$ be the set of clauses of a given instance~$\phi$.
    Define $G=(V,E)$ as follows:
    \begin{itemize}
        \item For each variable $x_i$, we introduces a path $x_i, y_i, z_i, \overline{y_i},\overline{x}_i$. We naturally interpret the vertices $x_i$ and $\overline{x}_i$ as literals.  Furthermore, add a pendant vertex~$w_i$ connected to~$z_i$, as well as three pendant vertices $x_{i,1},x_{i,2},x_{i,3}$ (or $\overline{x}_{i,1},\overline{x}_{i,2},\overline{x}_{i,3}$, respectively) connected to $x_i$ (or $\overline{x}_i$, respectively).
        \item For each clause $c_j$, add two vertices $c_j,c_j'$ together with an edge $\{c_j,c_j'\}$ and an edge between $c_j$ and each literal of this clause. 
        \item Introduce a cycle $u_1,\ldots,u_{2n+4m},u_1$ and add to each vertex $u_r$ three pendant vertices $u_{r,1},u_{r,2},u_{r,3}$. Furthermore, for each $i\in [n]$ and $j\in [m]$, add the edges $\{z_i,u_{2i}\}, \{c_j,u_{2n+4j-3}\}, \{c_j,u_{2n+4j-1}\}$. 
    \end{itemize}
    More formally, this results in the following graph:
    \begin{equation*}
        \begin{split}
            V=\,& \{c_j,c_j'\mid j\in [m]\}\cup \{w_i,x_i,x_{i,k},\overline{x}_{i},\overline{x}_{i,k}, y_i,\overline{y}_i,z_i\mid i\in [n],k\in [3]\} \cup{}\\
            &\{u_r, u_{r,k} \mid r\in [2n+4m],k\in [3]\}, \text{ and}\\
            E=\,&\{\{u_r, u_{r,k}\}, \{u_r,u_{(r\bmod (2n+4m))+1}\} \mid r\in [2n+4m],k\in [3]\}\cup{}\\
            &\{\{z_i,y_i\},\{x_i,y_i\},\{z_i,\overline{y}_i\},\{\overline{x}_i,\overline{y}_i\},\{u_{2i},z_i\},\{w_{i},z_i\}\mid i\in [n]\}\cup{}\\ 
            &\{\{{x}_i,{x}_{i,k}\},\{\overline{x}_i,\overline{x}_{i,k}\}\mid i\in [n],k\in [3]\}\cup{}\\
            &\{\{c_j,u_{2n+4j-3}\},\{c_j,u_{2n+4j-1}\}, \{c_j,l\}, \{c_j,c_j'\}\mid j\in [m], l\in c_j\}.
        \end{split}
    \end{equation*}   
    
    A sketch of this graph can be found in \autoref{fig:trdf_paranp_maxdeg3}. All vertices in that drawing show their true degrees but the clause vertices from $C$ that will have exactly three neighbors from the literal vertices. 
    Notice that the drawing of the cycle~$U$ formed by the vertices $u_r$ is slightly misleading, as ``first'' the $2n$ many vertices $u_{2i}$ adjacent to $z_i$ (and the neighbors of $u_{2i}$) appear on this cycle, followed by the vertices adjacent to $c_j$ (and neighbors) that form the ``last part'' of $4m$ vertices of this cycle.   Let $P$ be the set of pendant vertices, i.e., $P=\{c_j',w_i,x_{i,k},\overline{x}_{i,k},u_{r,k}\mid j\in[m], i\in[n],k\in[3],r\in[2n+4m]\}$. These vertices are drawn with small rectangles in \autoref{fig:trdf_paranp_maxdeg3}, while vertices of higher degree are drawn with small circles. 

After explaining the structure of the constructed graph~$G$, next we show that~$G$ indeed satisfies the claimed graph properties.

\shortversion{\begin{claim}\shortversion{$(\ast)$}
Graph~$G$ is bipartite, has maximum degree~6 and degeneracy~2.   
\end{claim}}

\begin{toappendix}
\subsection{Details of the proof of \autoref{thm:another_def_all_max_deg_6}}
\begin{claim}
Graph~$G$ is bipartite, has maximum degree~6 and degeneracy~2.   
\end{claim}
\begin{pfclaim}
    The graph~$G$ is bipartite with the partition classes $B_1=\{y_i,\overline{y}_i,x_{i,k},\overline{x}_{i,k},w_i\mid i\in [n]\} \cup \{u_{2r},u_{2r-1,k}\mid r\in [n+2m], k\in [3]\} \cup C$ and $B_2=\{x_i,\overline{x}_i,z_i\mid i\in [n]\} \cup \{u_{2r-1},u_{2r,k}\mid r\in [n+2m], k\in [3]\}\cup \{c_j'\mid j \in [m]\}$.
    We have drawn vertices of $B_2$ in black in \autoref{fig:trdf_paranp_maxdeg3} for clarity.

    Graph~$G$ has maximum degree~6; more specifically, the vertices of highest degree~6 are the literal vertices $x_i$ and $\overline{x}_i$ and the vertices $u_{2i}$, $c_j$, $u_{2n+4j-3}$ and $u_{2n+4j-1}$ (for $i\in [n]$ and $j\in[m]$) on the mentioned cycle.

    Furthermore, $G$ has degeneracy 2: After deleting pendant vertices and $y_i,\overline{y}_i$ (as $\deg(y_i)=\deg(\overline{y}_i)=2$) for $i\in [n]$, $z_i$ and $x_i,\overline{x}_i$ have at most $2$ neighbors and can be deleted. Then, all $c_j$ will have at most 2 neighbors. Finally, only the cycle formed by the vertices $u_r$ remains and can be removed. 
\end{pfclaim}
  
\end{toappendix}

To fully specify the \AnIMDefAll instance, we also have to provide a list of inclusion minimal defensive alliances. For this purpose, let $T$ be the set of connected components of $G[C\cup \{x_i,\overline{x}_i,y_i,\overline{y}_i \mid i\in [n]\}]$. 
Define $L\coloneqq\{\{v\}\mid v\in P\}\cup T \cup \{\{z_i,y_i,\overline{y}_i\}\mid i\in [n]\}$. By going through the inequalities, the following can be proven: 

\shortversion{\begin{claim} \shortversion{$(\ast)$}
$L$ is a list of inclusion minimal defensive alliances.   
\end{claim}}

\begin{toappendix}
\begin{claim} 
$L$ is a list of inclusion minimal defensive alliances.   
\end{claim}

\begin{pfclaim}
  Clearly, $\{\{v\}\mid v\in P\}$ contains only inclusion minimal defensive alliances.
    Moreover, the sets $\{z_i,y_i,\overline{y}_i\}$ form inclusion minimal defensive alliances for each $i\in [n]$ (Since $\deg(y_i)=\deg(\overline{y}_i)=2$ and $\deg(z_i)=4$, we can use arguments from \autoref{lem:max_deg_5_path__min_def_all}).

    Let us now discuss~$T$.
By monotonicity of~$\phi$, for each $A\in T$, either $A\subseteq C^+\cup \{x_i,y_i \mid i\in [n]\}$ or $A\subseteq C^-\cup \{\overline{x}_i,\overline{y}_i \mid i\in [n]\}$, where $C^+\subseteq C$ is the set of clauses that contains only unnegated variables and $C^-\subseteq C$ is the set of clauses that contains only negated variables.

     By symmetry, we discuss the first case only in the following. Observe that each $A\in T$ is a defensive alliance: each $y_i\in A$ has its neighbor~$z_i$ not in~$A$, but its other neighbor $x_i$ is in~$A$; each $x_i\in A$ has three pendant neighbors outside~$A$, but three other neighbors within~$A$, namely, $y_i$ and two clause vertices from $C^+$; each $c_j\in A\cap C^+$ has three variable vertices from $A$ in its neighborhood, and also three neighbors outside~$A$, namely, $c_j'$, $u_{2n+4j-1}$ and $u_{2n+4j-3}$. \longversion{Furthermore}\shortversion{Also}, every $A\in T$ is minimal: as can be seen by revisiting the case-by-case consideration of neighborhoods, removing any number of vertices from~$A$ in order to produce a proper subset~$A'$ would change the neighborhood of \longversion{at least one}\shortversion{some} remaining vertex $x\in A'$ such that $x$ is undefended, i.e., $\deg_{A'}(x)+1< \deg_{V\setminus A'}(x)$.  
   \end{pfclaim}
\end{toappendix}

     It is straightforward to see that the graph and also the list $L$ can be constructed in polynomial time from the formula~$\phi$. Finally, we show that the given $\NP$-hardness reduction actually works.  We briefly explain the idea of the proof: Let $D$ be a minimal defensive alliance on $G$. Remember that a set consisting only of a pendant vertex is a defensive alliance. Hence, a pendent vertex cannot be in a minimal defensive alliance with another vertex.  Because of the bounded degree, if a vertex $v$ with three pendant neighbors is in $D$, $N(v)\setminus P\subseteq D$. Therefore, if  there is a $r\in [2n+4m]$ with $u_r\in D$, $\{z_i\mid i\in [n]\}\cup C\cup \{u_s\mid s\in [2n+4m]\}\subseteq D$. For each $i\in [n]$ with $z_i\in D$, $\{y_i,z_i,\overline{y}_i\}\in L$ implies $u_{2i}\in D$ and $\vert \{y_i,\overline{y}_i\}\cap D\vert=1$. Since $x_i$ and $\overline{x}_i$ have three pendant neighbors, if $z_i\in D$ then only one of $x_i$ and $\overline{x}_i$ can be in $D$ and $\{z_i\mid i\in [n]\}\cup C\cup \{u_s\mid s\in [2n+4m]\}\subseteq D$. Let there be a $j\in [m]$ with $c_j\in D$. Since the connected component of $G[C\cup \{x_i,\overline{x}_i,y_i,\overline{y}_i \mid i\in [n]\}]$ including $c_j$ is in $L$, at least one literal in $c_j$ and one of $u_{2n+4j-3},u_{2n+4j-1}$ have to be inside of the solution. Hence, any minimal defensive alliance $D\notin L$ includes $\{z_i\mid i\in [n]\}\cup C\cup \{u_s\mid s\in [2n+4m]\}$ and, for each  $j\in [m]$, there exists a literal in $c_j$ that is also included in $D$. We obtain a satisfying assignment of~$\phi$ by setting the literals in $x_j$ to true. Details can be found in the following proof.
    
    \begin{claim}
        There is an inclusion minimal defensive alliance $A\notin L$ \iffl there is a satisfying assignment for $\phi$.
    \end{claim}
    \begin{pfclaim}
    Let us first show the if-direction. 
        Let $\alpha$ be a satisfying assignment. Define   \begin{equation*}\begin{split}A\coloneqq{}& \{u_r \mid r\in [2n+4m]\}\cup  C\cup \{z_i,y_i,x_i\mid i\in [n], \alpha(x_i)=1\}\cup{}\\& \{z_i,\overline{y}_i,\overline{x}_i\mid i\in [n], \alpha(x_i)=0\}.\end{split}\end{equation*}
        If there is an $l\in \{x_i, \overline{x}_i\mid i\in [n]\} \cap A$ such that for all $c_j \in C$ with $l \in c_j$, $c_j$ is satisfied by three literals then delete $l$ from $A$ (but only one per clause). Observe that for all $r\in [2n+4m]$, $\deg_{A}(u_r)+1\geq 3=\deg_{\overline{A}}(u_r)$. Let $i\in [n]$. Then $\deg_A(z_i)+1=3>2=\deg_{\overline{A}}(z_i)$. In the following, we consider $\alpha(x_i)=1$. The case $\alpha(x_i)=0$ follows analogously. Observe that $\deg_A(y_i)+1\geq 2 >1= \deg_{\overline{A}}(y_i)$. Furthermore, $\deg_A(x_i)+1 = 4>3=\deg_{\overline{A}}(x_i)$. For all $j\in [m]$, $c_j$ is satisfied and $\deg_A(c_j)+1\geq 4>3\geq \deg_{\overline{A}}(c_j)$. Thus, $A$ is a defensive alliance. Since we deleted $l\in \{x_i, \overline{x}_i\mid i\in [n]\} \cap A$ with for all $l\in c_j \in C$, $c_j$ is satisfied by three literals, there is no $A'\in L$ with $A' \subseteq A$. Hence, there is an inclusion minimal defensive alliance $ A$ not in $L$. 

\shortversion{The only-if direction can be found in the appendix.}
\begin{toappendix}
      \shortversion{\smallskip \noindent We now provide the only-if direction of the proof of the following claim: \begin{claim}
        There is an inclusion minimal defensive alliance $A\notin L$ if and only if there is a satisfying assignment for $\phi$.
    \end{claim}}
        \longversion{Conversely, l}\shortversion{L}et there be an inclusion minimal defensive alliance $D\subseteq V$ with $D\notin L$.  Observe that $P\cap D= \emptyset$. For all $r\in [2n+4m]$, $\deg_{D}(u_r)\leq \deg_{\overline{P}}(u_r)\leq 3 = \deg_P(u_r)\leq \deg_{\overline{D}}(u_r)$.
        Since $\deg_{\overline P}(u_r)\in \{2,3\}$, for $u_r\in D$, $N(u_r)\setminus P\subseteq D$. By an inductive argument, if $u_s\in D$ for some $s\in [2n+4m]$, then 
        \[\{z_i\mid i\in[n]\} \cup \{u_r\mid r\in [2n+4m]\} \cup C\subseteq D.\]
        Analogously to the case $u_r \in D$ with $r\in [2n+4m]$, for an $i\in [n]$ such that $x_i \in D$, $N(x_i)\setminus P\subseteq D$. Let $i\in [n]$. For $z_i\in D$, $\deg_D(z_i)\geq 2$. Since $\{z_i,y_i,\overline{y}_i\}\in L$, $z_i\in D$ implies $\{y_i,\overline{y}_i\}\nsubseteq D$ and $u_{2i}\in D$.

        Now further assume for $r\in [2n+4m]$ such that $u_r\notin D$. Let there be a $j\in [m]$ such that $c_j\in D$. Hence, the three literals must be in $D$. Without loss of generality let $x_{i_1},x_{i_2},x_{i_3}$ be the literals in $c_j$. As mentioned above, $N(x_{i_j})\setminus P\subseteq D$. By an inductive argument, $D$ includes a connected component of  $G[C\cup \{x_i,\overline{x}_i,y_i,\overline{y}_i \mid i\in [n]\}]$. Hence, $D$ is not minimal or $D$ is in $L$. Hence, if there is a $c_j\in D$, then $\{u_r\mid r\in [2n+4m]\}\subseteq D$. This implies \[\{z_i\mid i\in[n]\} \cup \{u_r\mid r\in [2n+4m]\} \cup C\subseteq D.\] Define $\phi:X\to\{0,1\}$ with $\phi^{-1}(1)\coloneqq\{x_i\mid i\in [n],x_i\in D\}$. Let $c_j\in C$; then, there is at least one literal $x'\in C_j\cap D$. If $x'$ is a positive literal, then $\phi$ satisfies~$c_j$. Therefore, we consider the case where $x'$ is negative ($x'=\overline{x}_i$). As mentioned above, $z_i\in D$ and $N(\overline{x}_i)\setminus P\subseteq D$. Hence, $\overline{y}_i\in D$. If $x_i\in D$, then $y_i\in D$. This would imply $\{z_i,y_i,\overline{y}_i\}\subseteq D$ which contradicts the minimality of~$D$. \shortversion{\hfill$\Diamond$}
\end{toappendix}
    \end{pfclaim}
    This concludes the proof of the claimed $\NP$-completeness result.\qed
\end{proof}

\section{Parameterized Enumeration Algorithms}\label{sec:paraenum}

Motivated by the proven non-existence of an output-polynomial time algorithm (assuming $\ptime\neq\NP$), we now  consider the parameterized complexity for enumerating  inclusion minimal defensive alliances. We consider the parameters neighborhood diversity, pathwidth and treewidth. We start by providing an \FPT-delay algorithm when parameterized by neighborhood diversity. Then, we consider treewidth and pathwidth. We prove that there is no \FPT-delay algorithm for these parameters, unless $\FPT = \W{1}$.
\longversion{\subsection{Neighborhood Diversity}}
\shortversion{Another}\longversion{One of the most innocently looking} structural graph parameter is\longversion{ probably} the \emph{neighborhood diversity} $\nd(G)$ of a graph~$G$. 
 This is the number of equivalence classes of the relation $u\sim_{\nd}v$ \iffl $N(u)\backslash\{v\}=N(v)\backslash\{u\}$.\longversion{ Obviously, $\nd(K_n)=1$ and $\nd(K_{n,n})=2$ if $n>1$.} Notice that, given~$G$, $\nd(G)$ can be computed in polynomial time. It is well-known that the equivalence classes of $\sim_{\nd}$ are either independent sets or cliques.

\begin{lemrep}\applabel{lem:add_nd_vertex}
    Let $G=(V,E)$ be a graph, and $C_1,\ldots,C_{\nd(G)}$ be the neighborhood classes of $G$ and $A\subseteq V$ a defensive alliance. For all $i\in [\nd(G)]$ with $C_i\cap A\neq \emptyset$ and $v\in C_i$, $A\cup \{v\}$ is a defensive alliance. 
\end{lemrep}

\begin{proof}
    Let $i\in [\nd(G)]$, $v\in C_i$ with $u\in C_i\cap A$ and $v\notin A$. For $A_v \coloneqq A \cup \{v\}$ and $w\in A$, $\deg_{A_v}(w) +1 \geq \deg_{A}(w)+1 \geq \deg_{\overline{A}}(w)\geq \deg_{\overline{A_v}}(w).$
    Furthermore, \[\deg_{{A_v}}(v) + 1 = \deg_{A_v}(u) +1 \geq \deg_{A}(u)+1 \geq \deg_{\overline{A}}(u) \geq \deg_{\overline{A_v}}(u) = \deg_{\overline{A_v}}(v).\]
    Hence, $A_v$ is a defensive alliance. \qed 
\end{proof}

\begin{lemma}\label{lem:nd_global_min_char}
    Let $G=(V,E)$ be a graph, $C_1,\ldots,C_{\nd(G)}$ be the neighborhood classes of $G$ and $A\subseteq V$ a defensive alliance. Then, $A$ is an inclusion minimal defensive alliance \iffl the following conditions hold:
    \begin{enumerate}
        \item for all $v\in A$, $A\setminus \{v\}$ is no defensive alliance and \label{con:nd_min_local}
        \item for\longversion{ all} $I \subsetneq \{i\in [\nd(G)] \mid C_i\cap A \neq \emptyset\}$, $A\cap \left( \bigcup_{j\in I} C_j\right)$ is no defensive alliance. \label{con:nd_min_classes}
    \end{enumerate}
\end{lemma}
The intuition of the second point is that, after deleting from $A$ any combination of neighborhood classes intersecting $A$, the resulting set is no longer a defensive alliance due to the minimality of~$A$. 
\begin{proof}
    Observe that all sets in the constraints are subsets of $A$. Hence if one would be a defensive alliance, $A$ would not be inclusion minimal. 

    For the other direction, assume the constraints hold but there is a defensive alliance $B\subsetneq A$. Define $I \coloneqq \{ i\in [\nd(G)] \mid C_i\cap B\}$. By using  \autoref{lem:add_nd_vertex} inductively, $A\cap \left(  \bigcup_{j\in I} C_j\right)$ is a defensive alliance. Condition~\ref{con:nd_min_classes} implies $I=\{i\in [\nd(G)] \mid C_i\cap A \neq \emptyset\}$. Consider some $v \in A \setminus B$. Observe that there is an $i\in [\nd(G)]$ with $v\in C_i \cap A$ and $B\cap C_i\neq \emptyset$. Again we can use  \autoref{lem:add_nd_vertex} to prove that $A\setminus \{v\}$ is a defensive alliance. This contradicts Condition~\ref{con:nd_min_local}. Hence, $A$ is minimal. \qed
\end{proof}

We will use \autoref{lem:nd_global_min_char} to build an \ilp for enumerating all inclusion minimal defensive alliances. To this end, we\longversion{ first} define the enumeration variant of \ilps (called \textsc{Enum-ILP}):
Given are $k,m\in \mathbb{N}$, $M\in \mathbb{Z}^{m\times k}$ and $b\in \mathbb{Z}^m$ such that the size of $L_{M,b}:=\{x\in \mathbb{Z}^k\mid Mx\leq b\}$ is bounded by $2^{h(k)p(s)}$, where $s$ is the size of the input, $h$ is a computable function and $p$ is some polynomial. The goal is to enumerate $L_{M,b}$, which can be done with the help of the following result.

\begin{proposition}[{\cite[Proposition~5]{CreMPU2026}}]\label{pro:ILP_FPT}
    $k\text-{\normalfont \textsc{Enum-ILP}}\in\DelayFPT$.
\end{proposition}

\shortversion{\noindent Next, we  consider enumerating defensive alliances parameterized by \nd.}
\longversion{In the following, we will consider the enumeration of defensive alliances, parameterized by neighborhood diversity.}

\begin{theorem}\label{thm:FPT-delay_nd}
    $\nd$-$\enum \DefAll\in \DelayFPT$.
\end{theorem}
Before proving the theorem, we need some auxiliary results first. 
Let $G=(V,E)$ be a graph with the neighborhood diversity classes $C=\{C_1,\ldots,C_{\nd(G)}\}$. Observe that, for all $v\in D\subseteq  V$, $\deg_D(v)+1\geq \deg_{\overline{D}}(v)$ is equivalent to $\deg(v)\leq 2 \deg_D(v)+1$\longversion{, since $\deg(v) = \deg_D(v)+\deg_{\overline{D}}(v)$}. We will use \shortversion{this}\longversion{the second equation} to construct the ILP. To simplify the ILP definition, we introduce some notation for $i\in [\nd(G)]$:
\begin{itemize}
    \item $N_i\coloneqq \{j\in [\nd(G)]\mid C_j\cap N(C_i)\neq \emptyset\}$ is the set of indices of neighbored classes (possibly including~$i$  when $C_i$ is a clique and $\vert C_i\vert \geq 2$).
    \item $a_i\in \{0,1\}$ satisfies $a_i=1$ \iffl $i \in N_i$.
    \item $d_i\in \mathbb{N}$ equals $\deg(v)$ for any $v\in C_i$.
\end{itemize}
The algorithm examines each\longversion{ subset} $A\subseteq  [\nd(G)]$.  $A$ should represent the set of classes intersecting the inclusion minimal defensive alliance we would like to enumerate. For each $A$ we will construct an \ilp. So, we fix an inclusion minimal defensive alliance $D\subseteq V$ and\longversion{ thus a set $A\subseteq [\nd(G)]$ with} $A=\{i\in [\nd(G)] \mid C_i \cap D\neq \emptyset\}$.

If the graph $H[A] \coloneqq (A,\{\{i,j\}\mid i,j\in A, i\in N_j\})$ (that might contain loops) is not connected, we can discard $A$\shortversion{ as then $A$ cannot correspond to any minimal alliance in~$G$}.\longversion{ Namely, a defensive alliance $D\subseteq V$ with $D\cap C_i\neq \emptyset$ for each $i \in A$ could not be minimal, as it would not be connected.} For $\vert A\vert =1$, it is easy to check if we encode a minimal defensive alliance, since the corresponding vertex set of~$A$ would be an independent set or a clique, and clearly, these alliances can be enumerated efficiently: Let $A=\{i\}$. First consider that $C_i$ is an independent set. Since an inclusion minimal defensive alliance has to be connected, the defensive alliance consists of only one vertex which needs to be pendant or isolated. Now, consider $C_i$ is a clique. Since each vertex in $C_i$ has the same closed neighborhood, we can substitute the vertices in $C_i$ arbitrarily. Let $D\subseteq C_i$ a defensive alliance and $v\in D$, then $\deg_D(v)+1\geq \deg_{\overline{D}}(v)= \deg(v)-\deg_{D}(v)$. Hence, $D\in \binom{C_i}{\lceil\frac{\deg(v)-1}{2}\rceil}$, to be a minimal defensive alliance. 
\longversion{

}In the following, we will always assume (without further mentioning) that any $A\subseteq [\nd(G)]$ that we process further satisfies that $H[A]$ is connected and that $\vert A\vert>1$.

Now, we explain\longversion{ the idea of} the \ilp. We have five types of variables: for $i\in A$ the four types of variables $x_i\in [ \vert C_i\vert],y_i,u_i,w_i\in \{0,1\}$ and for $i\in I\subseteq A$, $z_{i,I}\in \{0,1\}$.
For $i\in A$, $x_i$ represents $\vert D\cap C_i\vert $. The ideas of the remaining types are described in \autoref{obs:ILP_idea_var}. To simplify the notation, we define $\delta_i(I)\coloneqq 2\cdot\left(-a_i  + \sum_{j\in N_i\cap I}x_{j}\right)$ for $I\subseteq A$ for $I\subseteq A$. Observe that for each $i\in I\subseteq A$ and $v\in C_i$, $\delta_i(I)=2\deg_{D}(v)+1$.
We  introduce the inequalities for the \ilp (we call it $\text{ILP}_A$) in \autoref{fig:ilp-formulation}. The proof of the following observation can be done by  a  straightforward case distinction.

\begin{figure}[tb]
\longversion{
\begin{align}
         d_i&\leq 2\cdot\left(-a_i  + \sum_{j\in N_i\cap A}x_{j}\right)+1 &\forall i\in A\label{con:def_all}\\
         d_i+2&\leq 2\cdot\left(-a_i  + \sum_{j\in N_i\cap A}x_{j}\right)+ 1 + 2y_i&\forall i\in A\label{con:yi_id1}\\
         d_i+ \vert V\vert  (1-y_i) &\geq 2\cdot\left(-a_i  + \sum_{j\in N_i\cap A}x_{j}\right) &\forall i\in A\label{con:yi_id2}\\
         (1-w_{i})&\leq x_i - 1 \leq \vert V\vert \cdot (1-  w_{i}) &\forall i\in A\label{con:wi}\\
         2 \cdot u_{i}&\leq w_i + y_i \leq 1 + u_i &\forall i\in A\label{con:ui}\\
         1 + u_i&\leq \left(\sum_{j\in N_i \cap A}y_{j}\right) &\forall i\in A\label{con:ILP_local}\\
        d_i &\leq 2\cdot\left(-a_i  + \sum_{j\in N_i\cap I}x_{j}\right)+1 +\vert V\vert z_{i,I} &\forall  I\subseteq A, \forall i\in I\label{con:ziI_id1}\\
        d_{i}+  \vert V\vert (1- z_{i,I}) &> 2\cdot\left(-a_i  + \sum_{j\in N_i\cap I}x_{j}\right)+1  &\forall  I\subseteq A, \forall i\in I\label{con:ziI_id2}\\
        1 &\leq \sum_{i\in I} z_{i,I} &\forall I\subseteq A\label{con:ILP_class}
\end{align}}
\shortversion{\begin{align}
         d_i&\leq \delta_i(A)+1 &\forall i\in A\label{con:def_all}\\
         d_i+2&\leq \delta_i(A)+ 1 + 2y_i&\forall i\in A\label{con:yi_id1}\\
         d_i+ \vert V\vert  (1-y_i) &\geq \delta_i(A) &\forall i\in A\label{con:yi_id2}\\
         (1-w_{i})&\leq x_i - 1 \leq \vert V\vert \cdot (1-  w_{i}) &\forall i\in A\label{con:wi}\\
         2 \cdot u_{i}&\leq w_i + y_i \leq 1 + u_i &\forall i\in A\label{con:ui}\\
         1 + u_i&\leq \left(\sum_{j\in N_i \cap A}y_{j}\right) &\forall i\in A\label{con:ILP_local}\\
        d_i &\leq\delta_i(I)+1 +\vert V\vert z_{i,I} &\forall  I\subseteq A, \forall i\in I\label{con:ziI_id1}\\
        d_{i}+  \vert V\vert (1- z_{i,I}) &> \delta_i(I)+1  &\forall  I\subseteq A, \forall i\in I\label{con:ziI_id2}\\
        1 &\leq \sum_{i\in I} z_{i,I} &\forall I\subseteq A\label{con:ILP_class}
\end{align}}
    \caption{An ILP associated with $A\subseteq V$.}
    \label{fig:ilp-formulation}
\end{figure}

\begin{obsrep}\applabel{obs:ILP_idea_var}
    Let $((x_j)_{j\in A},(y_j)_{j\in A},(u_j)_{j\in A},(w_j)_{j\in A},(z_{i,I})_{I\subseteq A, i\in I})$ be a solution of $\text{ILP}_A$. The following properties hold.
\begin{enumerate}
    \item For $i\in A$, $y_i\in \{0,1\}$ is~1 \iffl $\delta_i(A)+1 \in\{ d_i, d_i + 1\}$.\label{con:yi_idea}
    \item For $i\in A$, $w_i\in \{0,1\}$ is~1 \iffl $x_i=1$.\label{con:wi_idea}
    \item For $i\in A$, $u_i\in \{0,1\}$ is~1 \iffl $w_i=1=y_i$.\label{con:ui_idea}
    \item \label{con:ziI_idea} For $I \subsetneq A$\longversion{ and}\shortversion{,} $i\in I$, we introduce $z_{i,I}\in \{0,1\}$ which is~0 \iffl \shortversion{ $\delta_i(I)+1 \geq  d_i$.}\longversion{\[\delta_i(I)+1 \geq  d_i.\]}
    \item The vector $((y_j)_{j\in A},(u_j)_{j\in A},(w_j)_{j\in A},(z_{i,I})_{I\subseteq A, i\in I})$ is uniquely determined by the vector $(x_j)_{j\in A}$. \label{con:identity}
\end{enumerate}
\end{obsrep}
\begin{proof}
First, we prove that Inequalities \eqref{con:yi_id1} and \eqref{con:yi_id2} ensure Property~\ref{con:yi_idea}.

Let $i\in A$ and $D \subseteq V$ a corresponding set, and $v\in C_i \cap D$. Assume, $y_i=1$. Together with Inequality \eqref{con:yi_id2}, this implies 
$d_i=d_i+ \vert V\vert  (1-y_i) \geq \delta_i(A) .$
By Inequality \eqref{con:yi_id1}, we obtain 
$d_i=d_i+ 2-2y_i\leq  \delta_i(A) +1.$
Now, assume $y_i=0$. By Inequality \eqref{con:yi_id1}, 
$ d_i+2 \leq \delta_i(A)+ 1 .$
Hence, Property~\ref{con:yi_idea} holds.

Since $x_i\geq 1$, Inequality \eqref{con:wi} implies Property~\ref{con:wi_idea}. By Inequality \eqref{con:ui}, Property~\ref{con:ui_idea} holds. Next, we consider the property for $z_{i,I}$. Let $i\in I \subsetneq  A$. Inequality \eqref{con:ziI_id1} implies that if $z_{i,I}=0$, then
$ d_i \leq \delta_i(I)+1.$
For $z_{i,I}=1$,  Inequality \eqref{con:ziI_id2} implies
$ d_i=d_{i}+  \vert V\vert (1- z_{i,I}) > \delta_i(I)+1.$
Now, we consider Property~\ref{con:identity}. Let the vector $x=(x_j)_{j\in A}$ be fixed. By the Properties~\ref{con:yi_idea},~\ref{con:wi_idea} and~\ref{con:ziI_idea}, the values of $((y_j)_{j\in A},(w_j)_{j\in A},(z_{i,I})_{I\subseteq A, i\in I})$ are uniquely defined by~$x$. Because of Property~\ref{con:ui_idea} and the fact the $((y_j)_{j\in A},(w_j)_{j\in A})$ are uniquely defined, the values of $(u_j)_{j\in A}$ are uniquely defined. \qed
\end{proof}

This observation makes it easier to understand the ideas behind the variables: $(u_j)_{j\in A}$, $(w_j)_{j\in A}$ are just auxiliary variables. To understand the connection between solutions of $\text{ILP}_A$ and minimal defensive alliances, it is important to grasp the idea of $(y_j)_{j\in A}$ and $(z_{i,I})_{I\subseteq A, i\in I}$. Both are related to \autoref{lem:nd_global_min_char}. While $y_j$ is connected to Condition \ref{con:nd_min_local}, $z_{i,I}$ is connected to Condition \ref{con:nd_min_classes}. 
Let $D\subseteq V$ be a minimal defensive alliance such that for each $j\in [\nd(G)]$, $C_j\cap D =\emptyset$ \iffl $j\in A$.  We start by considering $(y_j)_{j\in A}$. Let $j\in A$ and $v\in D\cap C_j$. Observe that $\delta_j(A)\in \{d_i,d_i+1\}$ implies that for all $u\in N(v)$, $D\setminus \{u\}$ is no longer a defensive alliance, as $2\deg_{D\setminus\{u\}}(v)+1 \leq d_i +1-2=\deg(v)-1$. Hence, for each $j\in A$, if $\sum_{i\in N_j}y_i\geq 1+u_j$, $D\setminus \{u\}$ is no longer a defensive alliance for any $u\in C_j\cap D$. The variable $u_j$ excludes that the variable $y_j$ would affect the inequality as there is only one vertex in $C_j \cap D$ left. To see the connection between $(z_{i,I})_{i\in I\subsetneq A}$ and \autoref{lem:add_nd_vertex}, Condition \ref{con:nd_min_classes}, let $i\in I\subseteq A$. Observe that $z_{i,I}=1$ \iffl $2\deg_{D\cap \bigcup_{i\in I}  C_i}(v)+1<\deg(v)$. Hence, if for $I\subsetneq A$ there is an $i\in I$ with  $z_{i,I}=1$ (or, put differently, $\sum_{i\in I} z_{i,I}\leq 1$), $D\cap \bigcup_{i\in I}  C_i$ is not a defensive alliance. So we only need to check this condition for each $I\subsetneq A$. 
  
Since the values of $x=(x_j)_{j\in A}$ define a solution of $\text{ILP}_A$ uniquely (see \autoref{obs:ILP_idea_var}), we will refer to $x$ itself as a solution of the $\text{ILP}_A$. In this sense, $\text{Sol}_A=\bigtimes_{j\in A} [\vert C_j\vert]$ is the solution space of $\text{ILP}_A$.
In the following, a set $D\subseteq V$ \emph{corresponds to} a vector $x\in \text{Sol}_A$  \iffl for all $j\in A$, $\vert D\cap C_j\vert=x_j$ and for all $i\in [\nd(G)]\setminus A$, $D\cap C_i=\emptyset$. Observe that, for all
$x\in \text{Sol}_A$, a corresponding set $D\subseteq V$, $i\in A$ and $v\in C_i \cap D$, $\deg_D(v)= \left(-a_i  + \sum_{j\in N_i\cap A} x_{j} \right)$ holds.
We say that $D\subseteq V$ \emph{fits to} $A\subseteq [\nd(G)]$ if
$D \cap C_j \neq \emptyset$ for all $j\in A$ and $D \cap C_i = \emptyset$ for all $i\in [\nd(G)]\setminus A$. Then, define $\phi_A(D)\coloneqq (\vert D\cap C_j\vert)_{j\in A}\in \text{Sol}_A$. This function maps each %
vertex subset that fits to~$A$
to the vector it corresponds to. The idea of the algorithm is that for each inclusion minimal defensive alliance~$D$, $\phi_A(D)$ is the solution to $\text{ILP}_A$ and we can enumerate the inverse image of $\phi_A$ with polynomial delay.

\begin{lemrep}\applabel{lem:nd_ILP_solution}
   If $D$ is an inclusion minimal defensive alliance that fits to~$A$, then 
    $\phi_A(D)$ 
    is a solution of $\text{ILP}_A$.
\end{lemrep}

\begin{proof}
    Let $D$ be an inclusion minimal defensive alliance such that for $i\in [\nd(G)]$, $D \cap C_i \neq \emptyset$ \iffl $i\in A$. Define $x_i^D \coloneqq \vert C_i \cap D\vert $ for $i\in A$, yielding the vector $x^D=(x_i^D)_{i\in A}$. Let the vectors $(y_i^D)_{i\in A},(w_i^D)_{i\in A},(u_i^D)_{i\in A}$ and $(z_{i,I}^D)_{I\subseteq A, i\in I}$ be defined by $x^D$ according to the reasoning behind \autoref{obs:ILP_idea_var}. Observe that with this, the inequalities \eqref{con:yi_id1}, \eqref{con:yi_id2}, \eqref{con:wi}, \eqref{con:ui}, \eqref{con:ziI_id1}, and \eqref{con:ziI_id2} hold. Since $D$ is a defensive alliance, for all $i\in A$ and $v\in D\cap C_i$, 
    $d_i=\deg(v)  \leq  2\deg_D(v) + 1 = \delta_i(A) + 1. $ 
    By Condition~\ref{con:nd_min_classes} of \autoref{lem:nd_global_min_char}, for all $I\subsetneq A$, $D_I \coloneqq D\cap (\bigcup_{i\in I} C_i)$ is no defensive alliance. Hence, for each $I \subsetneq A$, there is an $i_I \in I$ such that for all $v\in D_I\cap C_{i_I}$, $d_{i_I}=\deg(v) > 2\cdot\deg_{D_I}(v)+1= \delta_{i_I}(I)+1.$ Hence, $z_{i,I}=1$ and Inequality \eqref{con:ILP_class} holds.

    Now, we will prove Inequality \eqref{con:ILP_local}. Let $i\in A$ and $v\in D \cap C_i$ also holds.
    By \autoref{lem:nd_global_min_char} Condition~\ref{con:nd_min_local}, $D_v\coloneqq D\setminus \{v\}$ is no defensive alliance. Thus, there are $j \in A$ and $s\in D_v \cap C_j$ such that $d_j > 2\deg_{D_v}(s)+1$. Since $D$ is a defensive alliance,  $d_j \leq 2\deg_{D}(s)+1$ implies $s\in N(v)$ and $j\in N_i$. Thus, $2\deg_{D}(s)+1\in \{d_j,d_j+1\}$ and $y_j=1$ by \autoref{obs:ILP_idea_var}, Condition~\ref{con:yi_idea}. If $i=j$, then $2 \leq \vert D \cap C_i\vert = x_i$, which implies with \autoref{obs:ILP_idea_var} Properties~\ref{con:wi_idea}  and~\ref{con:ui_idea}: $w_i=0=u_i$. Hence, $1\leq \sum_{j\in N_i\cap  A}y_j$ (this is Inequality~\eqref{con:ILP_local} for $u_i=0$).\qed
    \end{proof}

    \begin{lemrep}\applabel{lem:nd_def-alliances-for-ILP-solution}
        Let $x\in\text{Sol}_A$ be a solution to $\text{ILP}_A$. Then, there is a inclusion minimal defensive alliance $D^x$ that corresponds to~$x$.
    \end{lemrep}
\begin{proof} 
As noted above, a solution $x\in\text{Sol}_A$   to $\text{ILP}_A$ uniquely determines the complete solution 
    $((y_j)_{j\in A},(u_j)_{j\in A},(w_j)_{j\in A},(z_{i,I})_{I\subseteq A, i\in I})$. Let $D^x$ be any set corresponding to~$x$. For all $i\in A$ and $v \in D^x \cap C_i$, \[ \deg(v)= d_i \leq \delta_j(A)+1= 2 \deg_{D^x}(v)+1 .\] 
    Therefore, $D^x$ is  a defensive alliance. To prove minimality, we use \autoref{lem:nd_global_min_char}. First, we will prove Condition~\ref{con:nd_min_local}. For this, let $i\in A$ and $v\in D^x \cap C_i$. By Inequality \eqref{con:ILP_local}, there is at least one $j\in N_i \cap A$ such that $y_j=1$. If $y_i=1=w_i$ (so $\vert D^x\cap C_i\vert=1$; see \autoref{lem:nd_global_min_char}), then $u_i=1$ and there is a $j\in( N_i\cap A)\setminus \{i\}$ with $y_j=1$. Hence, there is an $s\in (D^x \cap C_j)\setminus \{v\}$ with \[ 2\deg_{D^x}(s)+1 =  \delta_j(A)+1 \in \{d_j,d_j+1\}= \{\deg(s),\deg(s)+1\}. \] This implies that $\deg(s)\geq2(\deg_{D^x}(s)-1)+1=2\deg_{D^x\setminus\{v\}}(s)+1.$ Therefore, Constraint~\ref{con:nd_min_local} of \autoref{lem:nd_global_min_char} holds. 
    
    Let $I \subsetneq A$. For Condition~\ref{con:nd_min_classes} of \autoref{lem:nd_global_min_char}, it is enough to show that $D^x_I \coloneqq D^x\cap \bigcup_{i\in I}C_i$ is no defensive alliance. By Inequality \eqref{con:ILP_class}, there is a $j_I$ with $z_{j_I,I}=1$. \autoref{obs:ILP_idea_var} implies for all $v\in D_I^x \cap C_i$, $\deg(v)=d_i> \delta_i(I)+1 = 2 \deg_{D^x_I}(v)+1.$ Hence, $D_I^x$ is no defensive alliance. By \autoref{lem:nd_global_min_char}, $D^x$ is an inclusion minimal defensive alliance.\qed 
\end{proof}

\begin{correp}\applabel{cor:nd_enumerate_equiv_def_all}
    Let $G=(V,E)$ be a graph with the neighborhood diversity classes $C_1,\ldots,C_{\nd(G)}$ and $A \subseteq [\nd(G)]$. For each solution $x=(x_i)_{i\in A}\in \text{Sol}_A$ of $\text{ILP}_A$, we can enumerate $\phi_A^{-1}(x)$ with polynomial delay.  
\end{correp}

\begin{proof}
    By the proof of \autoref{lem:nd_ILP_solution}, any set $D \subseteq V$ that corresponds to~$x$ is in $\phi_A^{-1}(x)$. Hence, a depth-first-search enumerating $\bigcup_{i\in A} \binom{C_i}{x_i}$ (which can be done with polynomial delay) would suffice. \qed 
\end{proof}

\begin{proof}[of \autoref{thm:FPT-delay_nd}]
    Let $G=(V,E)$ be a graph with neighborhood diversity $p=\nd(G)$.
    Let $A\subseteq [p]$. We first treat the rather trivial case that $|A|=1$ (see previous page). Then, we build the auxiliary graph $H[A]$ and only proceed if it is connected. Then, we build the corresponding \ilp; 
    $\text{ILP}_A$ includes $6\cdot \vert A\vert + 2^{\vert A\vert}+ 2 \cdot \vert A\vert  \cdot 2^{\vert A \vert} \in \mathcal{O}\left(2^{p\log(p)}\right)$ inequalities over $\vert A \vert = p$ variables. Thus, $\text{ILP}_A$ can be constructed in \FPT-time and all solutions can be enumerated with \FPT-delay (see \autoref{pro:ILP_FPT}). By \autoref{cor:nd_enumerate_equiv_def_all}, for each solution $x$ we can enumerate $\phi_A^{-1}(x)$ with polynomial delay. Clearly, for each solution $x$ of $\text{ILP}_B$ and $x'$ of $\text{ILP}_{B'}$ with $B,B'\subseteq [p]$, $\phi_B^{-1}(x)\cap\phi_B'^{-1}(x') =\emptyset$. Since there are at most $2^p$ many sets $A\subseteq [p]$, \longversion{all minimal defensive alliances can be enumerated with \FPT-delay when parameterized by~\nd}\shortversion{the result follows}.\qed 
\end{proof}

\begin{toappendix}
\noindent
\autoref{thm:FPT-delay_nd} implies the following corollary:
\begin{corollary}
     All inclusion minimal defensive alliances can be enumerated with \FPT-delay when parameterized by the vertex cover number.
\end{corollary}
\end{toappendix}

\longversion{\subsection{Treewidth}}

The treewidth~$\tw$ is possibly the most prominent structural graph parameter\longversion{ that is defined as follows}.
\begin{toappendix}
\begin{definition} A \emph{tree decomposition} of a graph $G$ is a pair $\mathcal{T}=(T, \{X_t\}_{t\in V(T)})$, where $T$ is a tree whose node $t$ is assigned a vertex subset $X_t\subseteq V(G)$, called a bag, satisfying the following:
    \begin{itemize}
        \item [1)] $\bigcup_{t\in V(T)}X_t=V(G)$;
        \item [2)] for each edge $uv\in E(G)$, there is some node $t$ of $T$ such that $u\in X_t, v\in X_t$;
        \item [3)] for each $u\in V(G)$, the set $T_u=\{t\in V(T)\mid u\in X_t\}$ induces a %
        subtree of~$T$.
    \end{itemize}
\end{definition}

The \emph{width} of tree decomposition $\mathcal{T}$ is given by $\max_{t\in V(T)}|X_t|-1$. The \emph{treewidth} of a graph~$G$, denoted as $\tw(G)$, is the minimum width over all tree decompositions of~$G$.     
\end{toappendix}
However, given~$G$, $\tw(G)$ cannot be computed in polynomial time unless $\ptime=\NP$. Still, it can be computed in \FPT-time, when parameterized by $\tw$. Hence, we can still consider parameterized enumeration with this parameter.

\longversion{In this section}\shortversion{Now}, we prove that there is no \FPT-delay algorithm enumerating all inclusion minimal defensive alliances\longversion{ when}\shortversion{,} parameterized by treewidth, unless $\FPT=\W{1}$. 
For this, we use the following \shortversion{\W{1}-hard }problem\shortversion{; see~\cite{Sze2008}}: 

\problemdef{Minimum Maximum Outdegree} {(\MinMaxOut)}{An
undirected graph $G=(V,E)$, $w\colon E \to \mathbb{N}^+$ given in unary, $r\in \mathbb{N}$}
{Is there an orientation of the edges such that, for each vertex, the sum of the weights of outgoing edges
is at most $r$?}

\longversion{This problem is known to be \W{1}-hard \cite{Sze2008}.}\shortversion{\noindent} The proof of the following result is inspired by~\cite{BliWol2018}.

\begin{figure}[bt]
    	\centering
	\begin{tikzpicture}[transform shape,scale=0.95]
		      \tikzset{every node/.style={circle,minimum size=0.1cm}}
            
			\node[draw,rectangle] (ve11) at (-1,0.5) {};
			\node[draw,rectangle] (ve12) at (-1,0.9) {};
			\node[draw,rectangle] (ve13) at (-1,1.3) {};
			\node[draw,rectangle] (ve21) at (-4.3,-0.75) {};
			\node[draw,rectangle] (ve22) at (-3.9,-0.75) {};
			\node[draw,rectangle] (ve23) at (-3.5,-0.75) {};
			\node[draw,rectangle] (ve31) at (-5.7,-0.4) {};
			\node[draw,rectangle] (ve32) at (-5.7,0) {};
			\node[draw,rectangle] (ve33) at (-5.7,0.4) {};\node[draw,rectangle] (ue11) at (1,0.5) {};
			\node[draw,rectangle] (ue12) at (1,0.9) {};
			\node[draw,rectangle] (ue13) at (1,1.3) {};
			\node[draw,rectangle] (ue21) at (4.3,-0.75) {};
			\node[draw,rectangle] (ue22) at (3.9,-0.75) {};
			\node[draw,rectangle] (ue23) at (3.5,-0.75) {};
			\node[draw,rectangle] (ue31) at (5.7,-0.4) {};
			\node[draw,rectangle] (ue32) at (5.7,0) {};
			\node[draw,rectangle] (ue33) at (5.7,0.4) {};
			\node[draw,rectangle] (hu11) at (4.8,2.6) {};
			\node[draw,rectangle] (hu12) at (5.2,2.6) {};
			\node[draw,rectangle] (hu21) at (4.8,4.7) {};
			\node[draw,rectangle] (hu22) at (5.2,4.7) {};
			\node[draw,rectangle] (hv11) at (-4.8,2.6) {};
			\node[draw,rectangle] (hv12) at (-5.2,2.6) {};
			\node[draw,rectangle] (hv21) at (-4.8,4.7) {};
			\node[draw,rectangle] (hv22) at (-5.2,4.7) {};
			\node[draw,rectangle] (e') at (0,-0.5) {};
            
            \node[draw] (e) at (0,0) {};
            \node[draw,label={below:$e_v$}] (ev) at (-1,0) {};
            \node[draw,label={below:$e_u$}] (eu) at (1,0) {};
            \node[draw,label={below:$u_e^1$}] (ue1) at (2,0) {};
            \node[draw,label={below:$u_e^{2}$}] (ue2) at (3,0) {};
            \node[draw] (uewe) at (5,0) {};
            \node[draw,label={below:$v_e^{1}$}] (ve1) at (-2,0) {};
            \node[draw,label={below:$v_e^{2}$}] (ve2) at (-3,0) {};
            \node[draw] (vewe) at (-5,0) {};
            \node[draw,label={left:$v_e^{-1}$}] (ve-1) at (-5,2) {};
            \node[draw,label={left:$v_e^{-w(e)}$}] (ve-we) at (-5,1) {};
            \node[draw,label={right:$u_e^{-1}$}] (ue-1) at (5,2) {};
            \node[draw,label={right:$u_e^{-w(e)}$}] (ue-we) at (5,1) {};
            \node[draw,label={right:$h_u^{1}$}] (hu1) at (5,3.2) {};
            \node[draw,label={right:$h_u^{2r-2}$}] (hu2r-2) at (5,4.2) {};            
            \node[draw,label={left:$h_v^{1}$}] (hv1) at (-5,3.2) {};
            \node[draw,label={left:$h_v^{2r-2}$}] (hv2r-2) at (-5,4.2) {};

            \node[draw,label={above:$u$}] (u) at (3,1) {};
            \node[draw,label={above:$v$}] (v) at (-3,1) {};
            \node[draw,label={above:$q$}] (q) at (0,3) {};
            \node[draw,rectangle,label={above:$q^1$}] (q1) at (-0.75,4) {};
            \node[draw,rectangle,label={above:$q^z$}] (qz) at (0.75,4) {};

			\node[] (dots1) at (4,0) {$\cdots$};
			\node[] (dots2) at (-4,0) {$\cdots$};
			\node[] (dots3) at (-5,1.6) {$\vdots$};
			\node[] (dots4) at (5,1.6) {$\vdots$};
			\node[] (dots5) at (-5,3.8) {$\vdots$};
			\node[] (dots6) at (5,3.8) {$\vdots$};
			\node[] (dots7) at (0,4) {$\cdots$};
			\node[] (elabel) at (0.3,0.3) {$e$};
			\node[rectangle] (vewelabel) at (5,-0.4) {$u_e^{w(e)}$};
			\node[rectangle] (uewelabel) at (-5,-0.4) {$v_e^{w(e)}$};
            
    		\path (v) edge[-] (q);
    		\path (u) edge[-] (q);
    		\path (e) edge[-] (q);
    		\path (q1) edge[-] (q);
    		\path (qz) edge[-] (q);
    		\path (q) edge[-] (hv1);
    		\path (q) edge[-] (hv2r-2);
    		\path (q) edge[-] (hu1);
    		\path (q) edge[-] (hu2r-2);
            
    		\path (e) edge[-] (ev);
    		\path (e) edge[-] (eu);
    		\path (ve1) edge[-] (ev);
    		\path (ue1) edge[-] (eu);

    		\path (v) edge[-] (ve1);
    		\path (v) edge[-] (ve2);
    		\path (v) edge[-] (vewe);
    		\path (v) edge[-] (hv1);
    		\path (v) edge[-] (hv2r-2);
    		\path (v) edge[-] (ve-1);
    		\path (v) edge[-] (ve-we);
    		\path (ve2) edge[-] (ve1);
    		\path (ve2) edge[-] (dots2);
    		\path (vewe) edge[-] (dots2);

    		\path (ve1) edge[-] (ve11);
    		\path (ve1) edge[-] (ve12);
    		\path (ve1) edge[-] (ve13);
    		\path (ve2) edge[-] (ve21);
    		\path (ve2) edge[-] (ve22);
    		\path (ve2) edge[-] (ve23);
    		\path (vewe) edge[-] (ve31);
    		\path (vewe) edge[-] (ve32);
    		\path (vewe) edge[-] (ve33);
    		\path (ue1) edge[-] (ue11);
    		\path (ue1) edge[-] (ue12);
    		\path (ue1) edge[-] (ue13);
    		\path (ue2) edge[-] (ue21);
    		\path (ue2) edge[-] (ue22);
    		\path (ue2) edge[-] (ue23);
    		\path (uewe) edge[-] (ue31);
    		\path (uewe) edge[-] (ue32);
    		\path (uewe) edge[-] (ue33);

    		\path (u) edge[-] (ue1);
    		\path (u) edge[-] (ue2);
    		\path (u) edge[-] (uewe);
    		\path (u) edge[-] (ue-1);
    		\path (u) edge[-] (ue-we);
    		\path (u) edge[-] (hu1);
    		\path (u) edge[-] (hu2r-2);
    		\path (ue2) edge[-] (ue1);
    		\path (ue2) edge[-] (dots1);
    		\path (uewe) edge[-] (dots1);

    		\path (hu21) edge[-] (hu2r-2);
    		\path (hu22) edge[-] (hu2r-2);
    		\path (hu11) edge[-] (hu1);
    		\path (hu12) edge[-] (hu1);
    		\path (hv21) edge[-] (hv2r-2);
    		\path (hv22) edge[-] (hv2r-2);
    		\path (hv11) edge[-] (hv1);
    		\path (hv12) edge[-] (hv1);
    		\path (e) edge[-] (e');
    		\end{tikzpicture}
        \caption{ \label{fig:twAIMDA_W1_hard} Sketch of the construction of \autoref{lem:twAIMDA_W1_hard} for an edge $\{v,u\}=e\in E$.}
\end{figure} 

\begin{lemma}\label{lem:twAIMDA_W1_hard}
 $\twAnIMDefAll$ is \W{1}-hard even on graphs with degeneracy~2.
\end{lemma}
\begin{proof}
We use parts of the construction of \autoref{thm:another_def_all_max_deg_6}.  
Let $G=(V,E)$ be a graph, $w\colon E \to \mathbb{N}^+$ be a weight function given in unary, and $r\in \mathbb{N}^+$. \longversion{Observe that the problem}\shortversion{\MinMaxOut} is polynomial-time solvable for $r\in \{1,2\}$\longversion{ (guess and check the outgoing edges for every vertex). Therefor}\shortversion{. Henc}e, we can assume $r\geq 3$. Define $G'=(V',E')$ as follows\shortversion{; see \autoref{fig:twAIMDA_W1_hard}}.
\begin{itemize}
	\item Make a copy of all vertices in $V$, keeping their names.
	\item For each $v\in V$ and $i\in [2r-2]$, introduce three vertices $h^i_v,h^i_v{}',h^i_v{}''$. 
	\item For all $ v\in e\in E, j\in [w(e)]$, add the vertices $v^{j}_e,v^{j}_{e}{}' ,v^{j}_e{}'' ,v^{j}_e{}''', v^{-j}_e , e,e_v, e'$.
	\item Introduce $q$, $q^1,\ldots,q^z$, for $z={\vert V\vert(2r-1)+ \vert E\vert }$. 
	\item Connect $q$ to $q^{\ell},v,h^i_v,e$ by an edge for $\ell\in [z], v\in V,  i\in [2r-2], e\in E$.
    \item For each $v\in e\in E$ and $j\in [w(e)]$, introduce the edges \[\{v,v^j_e\},\{v,v^{-j}_e\},\{v^j_e,v^j_e{}'\},\{v^j_e,v^j_e{}''\},\{v^j_e,v^j_e{}'''\}, \{v^1_{e},e_v\}, \{{e},e_v\}, \{{e},{e}'\},\] as well as $\{v_e^i,v_e^{i+1}\}$ for $i\in[w(e)-1]$.
    \item For all $v\in V$ and $i\in [2r-2]$, introduce the edges $\{v,h_v^i\},\{h_v^i,h_v^i{}'\},\{h_v^i,h_v^i{}''\}$.
\end{itemize}
\longversion{The construction is visualized in \autoref{fig:twAIMDA_W1_hard}.}
As the weight function is given in unary, one can easily observe that the graph~$G'$ can be computed from $(G,w,r)$ in polynomial time.

\shortversion{\begin{claim}
    $\tw(G')\leq \tw(G)+5$.
\end{claim}}

\begin{pfclaim}
Let $\mathcal{T}=(T=(V_T,E_T),\{X_t\}_{t\in V_T})$ be a tree decomposition of $G$ such that $\tw(G)=\max_{t\in V_T} \vert X_t\vert-1$. In the following, we will explain how to construct a tree decomposition $\mathcal{T}'=(T'=(V_{T'},E_{T'}),\{X'_t\}_{t\in V_{T'}})$ for $G'$ of treewidth $\tw(G)+5=\max_{t\in V_{T'}} \vert X'_t\vert -1$. 
\begin{itemize}
    \item We start with $\mathcal{T}'=(T'=(V_{T'},E_{T'}),\{X'_t\}_{t\in V_{T'}})\coloneqq(T=(V_T,E_T),\{X_t\cup \{q\}\}_{t\in V_{T}})$. We will gradually modify this initial tree in the following.
    \item For $v\in V$, there is a $t_v\in V_{T}$ with $v\in X_{t_v}$. Add $t_v^{\ell}$ for all $\ell\in [2r+1]$, with the edge $\{t_v,t_v^{\ell}\}$ and $X_{t_v^{\ell}} = X_{t_v}\cup \{q,h^{\ell}_v,h^{\ell}_v{}'\}$. 
    \item Take an arbitrary $t\in V_T$ and add a path $t,t_1,\ldots,t_{\vert E\vert + \vert V\vert(2r-1)}$ with $X'_{t_i}\coloneqq\{q,q_i\}$.
    \item Let $\{v,u\}=e\in E$. There is a $t_e\in V_T$  with $v,u\in X_{t_e}$. For $x\in e$, $j\in [w(e)]$, and $k\in [w(e)-1]$ add $t_e',t_e^x,t_e^{x,j},t_e^{x,j}{}'$ to $V_{T'}$ with $X'_{t_e'}=X_{t_e} \cup \{q,e,e',e_v,e_u\},X'_{t_e^x}=X_{t_e} \cup \{q,e_x,x_e^1\},X'_{t_e^{x,k}}=X_{t_e} \cup \{q,v_e^k,v_e^{k+1}\},X'_{t_e^{x,w(e)}}=X_{t_e} \cup \{q,v_e^{w(e)}\}, $ and $X'_{t_e^{x,j}{}'}=X_{t_e} \cup \{q,v_e^k,v_e^k{}',v_e^k{}'',v_e^k{}'''\}$,  together with the edges $\{t_e,t_e'\},\{t_e^x,t_e'\},\{t_e^x,t_e^{x,1}\},\{t_e^{x,j},t_e^{x,j}{}'\},\{t_e^{x,k},t_e^{x,k+1}\}.$\claimqed
\end{itemize}\renewcommand{\claimqed}{}
\end{pfclaim}
By $P=\{h_{v}^i{}',h_{v}^i{}''\mid v\in V, \ell \in [2r-2]\} \cup \{e',v_e^j{}',v_e^j{}'',v_e^j{}''', v_e^{-j}\mid v\in e\in E, j\in [w(e)]\} \cup \{q^{\ell}\mid \ell \in [\vert V\vert (2r-1)+ \vert E\vert ]\}$, denote the set of pendant vertices. Observe that, for all $e\in E$, $w(e)\leq r$. Otherwise, it would be a trivial \no-instance. Hence, there are at most $\vert V\vert (1+ r (2\deg(v)+1))+ 3\vert E\vert + 1 + \vert V  \vert (2r-1)+\vert E\vert$ vertices. As we can assume that $r\leq \sum_{e\in E}w(e)$ (otherwise, it is a trivial \yes-instance) and \shortversion{by the unary number encodings}\longversion{the weight function is given in unary}, $G'$ can be constructed in polynomial time.
\shortversion{\begin{claim}\shortversion{$(\ast)$}
    $G'$ is 2-degenerate.
\end{claim}}

\begin{toappendix}
    \begin{claim}
    $G'$ is 2-degenerate.
\end{claim}
\begin{pfclaim}
    To see that $G'$ is 2-degenerate, consider the following deletion ordering:
\begin{itemize}
    \item Delete all vertices in $P$ (for all $p\in P$, $\deg(p)=1$) and for each $\{v,u\} = e\in E$ delete $e_v,e_u$ ($N(e_v)=\{v^1_e,e\}$ and $N(e_u)=\{u^1_e,e\}$).
    \item Now $e\in E$ has only the neighbor $q$. Therefore, we can delete $e$ from $G'$.
    \item Observe that now, for all $v\in e\in E$, $v_e^{w(e)}$ has only the neighbors $v,v_e^{w(e)-1}$. Hence, we can delete $v_e^{w(e)}$. After doing this, $v_e^{w(e)-1}$ has two neighbors. By an inductive argument, we can delete all vertices $v_e^{w(e)},\ldots,v_e^1$. 
    \item For each $v\in V$ and $i\in [2r-2]$, after deleting $h_v^i{}'$, $N(h_v^i)=\{q,v\}$. Thus, we can delete $h_v^i$.
    \item This leaves us with the graph $(\{q,v\mid v\in V\}, \{\{q,v\}\mid v\in V\})$, which is a star and therefore 1-degenerate. 
\end{itemize}
\end{pfclaim}
\end{toappendix}

Define the list of given inclusion minimal defensive alliances  \[L\coloneqq\{\{p\}\mid p\in P\} \cup \{\{e,e_v,e_u\}\mid \{v,u\}=e\in E\}. \]
Notice that also this list $L$ can be computed in polynomial time, given~$G$.
We will shortly explain the idea behind~$L$. To this end, we compare it with the construction of \autoref{thm:another_def_all_max_deg_6}. The idea of $q$ in $G'$ is the same as the idea of the vertices $u_1,\ldots,u_{2n+4m}$: If one of these vertices is a part of the solution, we have to include some other specific vertices (all not pendant neighbors). Furthermore, for $\{v,u\}=e\in E$, $e_v,e_u,e$ work like $z_i,y_i,\overline{y}_i$: First observe that $\{e,e_v,e_u\}$ is an inclusion minimal defensive alliance. By adding these defensive alliances to the list of given solutions, we ensure that if $e$ (resp. $z_i$) belongs to a defensive alliance, then we have to add  either $e_v$ or $e_u$ (resp. $y_i$ or $\overline{y}_i$), but not both.  By the degree and the number of pendant neighbors, if $v_e^j$ (for $j\in [w(e)]$) is included in a solution, $N(v_e^1)\setminus P$ also has to be in the solution. This implies that we cannot add both $v_e^1$ and $u_e^1$ to a solution at the same time. This is important to understand how a solution of the  \textsc{AnIMDefAll} instance can be tranformed into a solution of the \textsc{MinMaxOut} instance: $v_e^1$ is in a solution $D$ \iffl $e$ is oriented as $(u,v)$.  Observe that if $v_e^{i}\in D$  for $v\in e$ and $i\in [w(e)]$, and a solution $D$, $\{v_e^{j}\mid j\in [w(e)]\}\subseteq D$, as each vertex in $\{v_e^{j}\mid j\in [w(e)]\}$ has degree at most~6 but 3 pendant neighbors. This mechanic simulates the weight of an edge. The number of vertices in $\{h^{k}_v\mid [k\in [2r-2]]\}$ and $\{v_{e}^{-i}\mid i\in [w(e)]\}$ is chosen such that the outdegree of $v$ is at most $r$ \iffl $\deg_D(v)+1\geq \deg_{\overline{D}}(v).$  

More details on this can be found in the proof of the following claim:

\shortversion{\begin{claim}\shortversion{$(\ast)$}
    $(G,w,r)$ is a \yes-instance of \textsc{MinMaxOut} \iffl there is an inclusion minimal defensive alliance in $G'$ that is not in~$L$.\qed
\end{claim}}

\begin{toappendix}
\begin{claim}
    $(G,w,r)$ is a \yes-instance of \textsc{MinMaxOut} \iffl there is an inclusion minimal defensive alliance in $G'$ that is not in~$L'$.
\end{claim}
\begin{pfclaim}
    Let $O=(V,E_O)$ be a solution for the \textsc{MinMaxOut} instance $(G,w,r)$. Define 
    \begin{equation*}
        \begin{split}
            D_O=\{q\} \cup V\cup E \cup \{h_v^i\mid v\in V, i\in [2r-2]\} \cup \left(\bigcup_{(u,v)\in E_O}\{e_v,v_e^j \mid j\in [w(e)]\} \right).
        \end{split}
    \end{equation*}
    Observe that no set in $L$ is a subset of $D$. Hence, it is enough to prove that $D_O$ is a defensive alliance. First, for any $e\in E$, $\deg_{D_O}(e)+1 =3>2=\deg_{\overline{D}}(e)$ and \[\deg_{D_O}(q)+1=\vert V\vert +\vert E\vert +\vert V\vert (2r-2)+1>  \vert E\vert +\vert V\vert (2r-1)=\deg_{\overline{D_O}}(q).\] Let $v\in V$. For $(u,v)\in E_O$ and $j\in [w(e)]$,  $N(e_v) \subseteq D_O$ and $\deg_{\overline{D_O}}(v_{e}^j)=3\leq\deg_{D_O}(v_e^j)$. 
    For $i\in [2r-2]$, $\deg_{D_O}(h_v^i) +1= 3>2=\deg_{\overline{D_O}}(h_v^i)$. This leaves us to consider~$v$: 
    \begin{equation*}
        \begin{split}
        \deg_{D_O}(v)+1&= 2r-2+1+1+\sum_{(u,v)\in E_O} w(\{u,v\})\\
        & \geq 2 \sum_{(v,u)\in E_O} w(\{u,v\}) + \sum_{(u,v)\in E_O} w(\{u,v\})= \deg_{\overline{D_O}}(v).
        \end{split}
    \end{equation*}
    Hence, $D$ is a defensive alliance and there is an inclusion minimal defensive alliance $D'\subseteq D$ not included in $L$.

    Conversely, let $D$ be an inclusion minimal defensive alliance with $D\notin L$. Hence $P\cap D= \emptyset.$ If $q\in D$, $\deg_P(q) = \vert E \vert + \vert V\vert (2r-1) = \deg(q)$ implies $N[q]\setminus P \subseteq D$. Analogously, for all $v\in V$ and $x\in \{h_v^i\mid i\in [2r-2]\} \cup \{v_e^{j}\mid v\in e\in E, j\in [w(e)]\}$, $x\in D$ would imply $N[x]\setminus P \subseteq D$. By an inductive argument, for $v\in e \in E$, if $D\cap \{v_e^j\mid j\in [w(e)]\}\neq \emptyset$, then $ \{v,e_v\}\cup \{v_e^j\mid j\in [w(e)]\} \subseteq D$. Let $\{v,u\}=e\in E$. Since $N(e)= \{q,e',e_v, e_u\}\subseteq \{q,e_v, e_u\}\cup P$, at least two of the vertices of $q,e_v,e_u$ must be in $D$. The case $\{e_v,e_u\}\subseteq D$ is impossible as otherwise, $\{e,e_v,e_u\}\in L $ would be a subset of $D$. Therefore, $e\in D$ implies $q\in D$ and so $N[q]\setminus P\subseteq D$. 
    
    Assume there is a $v\in V\cap D$ but $q\notin D$. Since $q\in N(h_v^i)\setminus P$ for each $i\in [2r-2]$, $\{h_v^i\mid i\in [2r-2]\}\cap D= \emptyset$. Thus, \[\deg_{\overline{D}}(v) \geq 2r-2+ \sum_{v\in e\in E} w(e)> \sum_{v\in e\in E} w(e)\geq \deg_{D}(v).\] This would contradict the fact that $D$ is a defensive alliance. Therefore, $V\cap D\neq \emptyset$ implies $q\in D$. 
    Let $v\in e\in E$. If $e_v\in D$ then $v_e^1\in D$\longversion{ (which implies $v\in D$)} or $e\in D$. Both cases imply $q\in D$.  In summary: if there is an inclusion minimal defensive alliance in $G'$ that is not in $L$, then $q\in D$. 

    Define the orientation $O_D=(V,E_D)$ with $E_D\coloneqq \{(u,v)\mid e = \{u,v\} \in E , e_v\in D\}$. Let $e=\{v,u\}\in E$. Since $q\in D$, $e\in E\subseteq  N[q] \setminus P$ is in $D$. Thus, either $e_v$ or $e_u$ have to be in the solution. Hence, $O_D$ is an orientation of $G$. If there is a $v\in V$ with $\deg^-_{O_D}(v)>r$, $D$ would not be a defensive alliance since:
    \begin{equation*}
        \begin{split}
            &\deg_D(v) +1 -\deg_{\overline{D}}(v)\\ 
            ={}&1  + 2r-2  + \left(\sum_{(u,v) \in E_D}  w(\{u,v\})\right) +1 - \left(\sum_{(v,u) \in E_D}  w(\{u,v\})\right) -\left(\sum_{v\in e \in E}  w(e)\right)\\
            ={}& 2r  - 2\left(\sum_{(v,u) \in E_D}  w(\{u,v\})\right)<2r-2r=0 
        \end{split}
    \end{equation*}
    This completes the proof of the claim\longversion{ and hence of the theorem}.
\end{pfclaim}
\end{toappendix}
\longversion{\qed}
\end{proof}

\noindent
The following can now be argued similar to \autoref{thm:another_def_all_max_deg_6}.

\begin{thmrep}\applabel{thm:no_tw_FPT_enum}
    $\twAnIMDefAll\notin\OutputFPT$ unless $\FPT= \W{1}$, even on graphs of degeneracy~2.  
\end{thmrep}

\begin{proof}
    Assume there is an output-\FPT algorithm $\mathcal{A}$ enumerating all inclusion minimal defensive alliances parameterized by $\tw$ in $f(\tw) ( n + N ) ^c$ (where $f:\mathbb{N}\to \mathbb{N}$ is a computable function, $n$ is the size of the input, $N$ is the number of outputs, and $c\in \mathbb{N}$). 

    Let $G$ be a graph and $L$ be a list of inclusion minimal defensive alliances of~$G$. The idea is to run the algorithm~$\mathcal{A}$ for $f(\tw) (n+\vert L\vert +2)^c$ steps. There are two cases:
    \begin{enumerate}
        \item The algorithm $\mathcal{A}$ still runs. Thus, $\mathcal{A}$ generated $\vert L \vert + 1$ outputs and $(G,L)$ is a \yes-instance of \textsc{AnIMDefAll}.
        \item The algorithm stops before $f(\tw) (n+\vert L\vert +2)^c$ steps. In this case, we compare all outputs with the list $L$. If there is one output not in $L$, then it is a \yes-instance. Otherwise, $\mathcal{A}$ enumerated all inclusion minimal defensive alliances of $G$ and each was in $L$. Therefore, it would be a \no-instance.
    \end{enumerate}
    Observe that the size of each inclusion minimal defensive alliance is bounded by $n$. Hence the running time is bounded by  $f(\tw) (n+L+2)^{(c+2)}$, which is \FPT. By \autoref{lem:twAIMDA_W1_hard}, such an algorithm cannot exist unless $\FPT=\W{1}$.\qed 
\end{proof}

\begin{remark}
    There are problems for which it is known that there is no output-\FPT algorithm unless $\FPT= \W{1}$. For these problems, it is often $\W{1}$-hard to find the first solution (for example \textsc{MaxOnes-SAT} \cite{CreMMSV2017}). \shortversion{As far as we know}\longversion{To the best of our knowledge}, \shortversion{this}\longversion{\autoref{lem:twAIMDA_W1_hard}} is the first use of a parameterized version of an \textsc{Another}-problem variation\longversion{ to obtain such a result}. 
\end{remark}

We can \longversion{actually }sharpen the previous considerations by using the parameter pathwidth instead of treewidth. \shortversion{This is the minimum width of a tree decomposition when the tree is a path.}
\begin{toappendix}
  \begin{definition} A \emph{path decomposition} of a graph $G$ is a tuple $\mathcal{P}=(X_1,\ldots,X_p)$, where for each $i\in [p]$, there is a is assigned a vertex subset $X_t\subseteq V(G)$, called a bag, satisfying the following:
    \begin{itemize}
        \item [1)] $\bigcup_{i\in [p]}X_i=V(G)$;
        \item [2)] for each edge $uv\in E(G)$, there is some\longversion{ index} $i\in [p]$ such that $u\in X_i, v\in X_i$;
        \item [3)] For all $i,j,k\in [p]$ with $i\leq j\leq k$, $X_i \cap X_k \subseteq X_j$.
    \end{itemize}
 The \emph{width} of a path decomposition $\mathcal{P}=(X_1,\ldots,X_p)$ is $\max_{i\in p}\vert X_i \vert - 1$. The \emph{pathwidth} of a graph $G$ is minimum width of a path decomposition on $G$ and is denoted by $\pw(G)$.
\end{definition}  
\end{toappendix}
\longversion{To this end}\shortversion{For this}, we \longversion{have to }re-analyze some of the previous\shortversion{ arguments and the \W{1}-hardness proof in~\cite{Sze2008}}\longversion{ly presented reductions and algorithms}.
\longversion{We will sketch in the following why our claims hold.}

\begin{thmrep}\applabel{thm:pw-MinMaxOut}
\pw-\MinMaxOut is \W{1}-hard.     
\end{thmrep}

\begin{proof} 
(Sketch)
Szeider has shown that   \textsc{MinMaxOut} is \W{1}-hard when parameterized by treewidth by giving a reduction from \textsc{Multicolored Clique}, or \textsc{MC}. We now describe a more detailed analysis of his proof.  In the graph~$H$ that he constructs from the \textsc{MC} instance that consists of $k$ color classes, (i.e., the reduction starts from a $k$-partite graph $G=(V_1,\ldots,V_k;E)$) two types of small trees emerge after deleting the $2\binom{k}{2}$ vertices of the form $b_{i,i'}$ and  $c_{i,i'}$ for $i,i'\in [k]$, $i<i'$, that form the set~$K$: 
\begin{itemize}
    \item For each color~$i\in[k]$, there is a vertex $a_i$ that is connected to vertices $u_i^j$, $j\in [n_i]$ representing the $n_i$ vertices in color class~$i$; moreover, each $u_i^j$ is adjacent to $x_i^j$ and to $y_i^j$.
    \item For each pair $(i,i')$ of color classes, with $i<i'$, there is a  vertex $d_{i,i'}$.
    For each edge $\{q,q'\}\in E$ of the $\textsc{MC}$ instance, where $q\in V_i$ and $q'\in V_{i'}$, there is an edge to $d_{i,i'}$.
\end{itemize}
These trees all have a very small height when rooted at $a_i$ or $d_{i,i'}$, respectively.
Hence, we can build a path decomposition of~$H$ of width $2\binom{k}{2}+2$, by walking through the trees one by one. More precisely, the first bag would consist of $K\cup\{a_1,u_1^1,x_1^1\}$, the second one of $K\cup\{a_1,u_1^1,y_1^1\}$, the third one of $K\cup\{a_1,u_1^2,x_1^2\}$, etc.  \qed
\end{proof}

\begin{toappendix}
With a small change of the construction sketched above, we can even arrive at a further parameterized hardness result for \MinMaxOut that might be of independent interest. To this end, let $\pthree(G)$ denote the smallest number of vertices that need to be deleted from~$G$ such that the resulting graph is a collection of paths of length at most~2. Looking at the proof of \autoref{thm:pw-MinMaxOut}, one observes that one can (further) remove all vertices $a_i$ and all vertices $d_{i,i'}$ from the graph. These are at most $\binom{k}{2}+k$ many. Hence, for the constructed graph~$H$, we understand that $\pthree(H)\leq 3\binom{k}{2}+k\in \mathcal{O}(k^2)$. Therefore, Szeider's reduction idea also proves:
\begin{theorem}
   \pthree-\MinMaxOut is \W{1}-hard.      
\end{theorem}
\end{toappendix}

Similarly, we can observe that when we assume that we start with a path decomposition for the \MinMaxOut instance, then the tree decomposition resulting from the proof of \autoref{lem:twAIMDA_W1_hard} is indeed a caterpillar that can be again converted into a path decomposition by ``intercalating'' the leaf vertices of that decomposition into a sequence; more precisely, if $B$ is a bag and $B_1,\dots,B_k$ are leaf bags connected to~$B$, then the sequence $B,B_1,\dots,B_k$
can replace that part as $B\subset B_i$ for all~$i$.  We can hence strengthen \autoref{lem:twAIMDA_W1_hard} and thus also \autoref{thm:no_tw_FPT_enum} towards a hardness result concerning path decompositions. This reasoning leads us to the next result.

\begin{theorem}\label{thm:no_pw_FPT_enum}
    $\pw\text-\AnIMDefAll\notin\OutputFPT$ unless $\FPT=\W{1}$, even on graphs of degeneracy~2. 
\end{theorem}

\section{Conclusions}

In this paper, we have contributed to the enumeration theory of alliance problems, which exhibits certain features that separate them from most other graph parameters. 
One of such properties was the non-monotonicity which cannot be expressed in monadic second order logic and therefore escapes many meta-theorems. 
In fact, alliance theory is by now a quite developed area, as also certified by surveys and book chapters; see  \cite{FerRod2014a,HayHedHen2021,OuaSliTar2018,YerRod2017}.  There are many variations of these notions, including offensive and powerful alliances. 
It would be interesting to look into these notions from the perspective of enumeration algorithms. 

A future research direction would be to find an \XP-delay algorithm for enumerating inclusion minimal defensive alliances when parameterized by \tw{} or~\pw. It could also be the case that there does not exist such an algorithm. Then we would like to find a proof for this. Considering other parameters like vertex integrity or twin-cover number would also be a interesting direction.

One of the standard techniques to obtain efficient enumeration algorithms is using so-called extension algorithms; see \longversion{\cite{BorGKM2000,CasFGMS2022,Mar2013a}}\shortversion{\cite{CasFGMS2022}}. In the context of defensive alliances, this was also investigated by Feng~\emph{et~al.}~\cite{FenFerMan2026}. We can interpret some of our reductions and algorithms also in this direction\shortversion{, as detailed in the appendix}.

\begin{toappendix}
    Let us first formally state the problem definition. Here, we go back to the notion of globally minimal alliances in order to stay consistent with the statement of the results of Feng~\emph{et~al.}~\cite{FenFerMan2026} that we are now going to strengthen.

\problemdef{Extension Inclusion Minimal Defensive Alliance} {(\ExtGMDefAll)}{A graph $G=(V,E)$, and sets $U,N\subseteq V$.}{Is there a globally minimal defensive alliance $A\subseteq V\setminus N$ in $G$ with $U\subseteq A$?}

\begin{theorem}\label{thm:NP_ExtGMAD}
    \ExtGMDefAll is \NP-complete, even on bipartite graphs of degree at most~6 and degeneracy~2. 
\end{theorem}

In~\cite[Theorem~2]{FenFerMan2026}, \NP-hardness of \ExtGMDefAll was only shown for bipartite graphs of degree at most~9.

\begin{proof}
    Revisit the proof of \autoref{thm:another_def_all_max_deg_6}.
    We construct the same graph from a \textsc{Monotone 3-Sat-(2,2)} instance~$\phi$ as before. Now, we consider the pre-solution $U=\{u_r\mid r\in [2n+4m]\}$. According to the arguments given in \autoref{thm:another_def_all_max_deg_6}, $U$ is contained in an inclusion minimal defensive alliance \iffl $\phi$ is satisfiable.\qed 
\end{proof}

One could be also interested in aspects of parameterized complexity, also see \cite{CasFGMS2022}. A standard parameter in this context is the size of the given pre-solution. This is also what we consider here. But one can also look at structural parameters.

\begin{theorem}\label{thm:W1_tw_ExtGMAD}
    \longversion{The problem }\ExtGMDefAll is \paraNP-complete, when parameterized by the size of the pre-solution. 
    \pw-\ExtGMDefAll is \W{1}-hard, even on graphs with degeneracy~2. 
\end{theorem}

\begin{proof}
    Revisit the proof of \autoref{lem:twAIMDA_W1_hard}.
    We construct the same graph from a \MinMaxOut instance~$(G,w,r)$ as before. Now, we consider the pre-solution $U=\{q\}$. According to the arguments given in \autoref{lem:twAIMDA_W1_hard}, $U$ is contained in an inclusion minimal defensive alliance \iffl $\phi$ is satisfiable. Also check the arguments leading to \autoref{thm:no_pw_FPT_enum}. \qed 
\end{proof}

Finally, one could ask for algorithmic results for extension problems. So far, only a polynomial-time algorithm for trees has been exhibited. This complements in fact the \paraNP-hardness result of the previously stated result. But we can also see some further algorithmic results for neighborhood diversity.

\begin{theorem}\label{thm:FPT-nd_ExtGMAD}
   \nd-\ExtGMDefAll is in \FPT. 
\end{theorem}

\begin{proof}
Recall how we have shown \autoref{thm:FPT-delay_nd}: we have reduced it to ILP enumeration by constructing an ILP (see \autoref{fig:ilp-formulation}) that has a solution \iffl there exists an inclusion minimal defensive alliance that is consistent with the chosen set of neighborhood equivalence classes. By adding a few inequalities, giving lower bounds on $x_i$ according to the given pre-solution, it is rather straightforward to employ the same ILP in order to decide if the pre-solution can be extended to an inclusion minimal defensive alliance. As the number of variables of the ILP is upper-bounded by a function of the neighborhood diversity of the given graph, the claim follows with~\cite{Len83}; also confer \cite{CygFKLMPPS2015}.
\qed 
\end{proof}

With this result, one could also avoid using ILP enumeration in order to deduce \autoref{thm:FPT-delay_nd}.

Finally, it would be very interesting to know if there is any form of \XP enumeration algorithm possible with the parameter treewidth or also pathwidth. This would then be the first problem for which we know of a  parameterized enumeration algorithm of \XP-type but can prove that under standard complexity assumption, no ouput\,\FPT enumeration algorithm could exist.

\end{toappendix}

\paragraph{Acknowledgements.} We would like to thank the reviewers of our SOFSEM submission, whose remarks helped improve the presentation of this paper.

\bibliographystyle{splncs04}
\bibliography{ab,DefAllEnum}

\end{document}